%% file: main_plain.tex
\documentclass{article}
\usepackage[margin=1in]{geometry}

\usepackage[utf8]{inputenc} 
\usepackage[T1]{fontenc}    
\usepackage{hyperref}       
\usepackage{url}            
\usepackage{bm}      
\usepackage{booktabs}       
\usepackage{amsfonts}       
\usepackage{nicefrac}       
\usepackage{microtype}      
\usepackage{amsmath}
\usepackage{microtype}
\usepackage{graphicx}
\usepackage{subcaption}
\usepackage{booktabs} 
\usepackage{float}
\usepackage[algo2e,ruled,vlined,linesnumbered]{algorithm2e}
\usepackage{comment}
\usepackage{pythonhighlight}
\usepackage{mdframed}

\newtheorem{theorem}{Theorem}[section]

\newtheorem{proposition}[theorem]{Proposition}
\newtheorem{remark}[theorem]{Remark}
\newenvironment{proof}{{\bf Proof:} }{\hfill\rule{2.1mm}{2.1mm}}

\usepackage{cleveref}

\newcommand{\mbb}{\mathbb}
\newcommand{\mbf}{\mathbf}
\newcommand{\mcl}{\mathcal}
\newcommand{\bs}{\boldsymbol}

\newcommand{\f}{\frac}

\newcommand{\T}{\textnormal}
\newcommand{\x}{\mathbf{x}}
\newcommand{\X}{\bm{\mathcal{X}}}
\newcommand{\D}{\bm{\mathcal{D}}}
\newcommand{\upsilona}[2]{\upsilon_{#1,#2}}

\usepackage[normalem]{ulem} 
\newcommand{\argmax}{\operatornamewithlimits{arg\,max}}

\usepackage[version=4]{mhchem}
\input{ex_shared_plain}

\ifpdf
\hypersetup{
  pdftitle={Lookahead Acquisition Functions for Finite-Budget Time-Dependent Bayesian Optimization: Application to Quantum Optimal Control},
  pdfauthor={}
}
\fi

\graphicspath{{figures/}}
\begin{document}

\maketitle

\begin{abstract}
We propose a novel Bayesian method to solve the maximization of a
time-dependent expensive-to-evaluate stochastic oracle. We are interested in the decision
that maximizes the oracle at a finite time horizon, given a limited budget of
noisy
evaluations of the oracle that can be performed before the horizon. Our recursive two-step 
lookahead acquisition function for Bayesian optimization makes
nonmyopic decisions at every stage by maximizing the expected utility at the specified time horizon. Specifically, we propose a generalized two-step lookahead framework with a customizable \emph{value} function that allows users to define the utility. We illustrate how lookahead versions of classic acquisition functions such as the expected improvement, probability of improvement, and upper confidence bound can be obtained with this framework. We demonstrate the utility of our proposed approach on several carefully constructed synthetic cases and a real-world quantum optimal control problem.
\end{abstract}



\section{Introduction}
\label{s:Introduction}
We consider the maximization of an expensive-to-evaluate oracle $f(\x, t)$,
where $\x\in \X \subset \mbb{R}^d$ are the inputs or \emph{action} parameters
in a compact domain and $t$ is the \emph{context} from a (possibly infinite)
set of contexts $\mcl{T}$. An observation $\hat{y} \in \mbb{R}$ made by observing
$f$ at the action-context pair $(\x, t)$ is assumed to be corrupted by
additive stochastic noise $\epsilon$; that is, $\hat{y} = f(\x,t) + \epsilon$.
Furthermore, we assume a context space $\mcl{T}$ whose members have a unique
ordering (hereafter ``time''). We seek an action $\x^*_T$ at a given future time $T \in \mcl{T}$ that solves
\begin{equation}
\argmax_{\x \in \X} f(\x, T)
\label{eq:prob}    
\end{equation}
in as few stochastic observations of $f$ as possible.
Such a problem fundamentally differs from general
context-dependent bandit problems \cite{krause2011contextual,srinivas2009gaussian} in that we are interested solely in
the action that maximizes the oracle value $f$ 
at time $t=T$, as opposed to maximizing the
\emph{cumulative oracle value} until $T$. Additionally, we
assume that only one observation can be made per time 
$t$.
The time-dependent maximization that we address arises, for example, in quantum
computing and control~\cite{magann2020digital,zhu2018training}, where 
values for the control parameter $\x$ are sought by using
noisy observations $f(\x,t)$ 
to achieve a
desired task at future time $T$; see \cref{ss:qoc}. Another application  relevant to our work
is multistage stochastic programming for portfolio optimization \cite[Ch. 1]{shapiro2014lectures}, where the goal is to find a 
policy $\{(\x_i, t_i) :~i=0,\ldots,n-1\}$ for redistributing assets 
to maximize the expected profit at $T=t_n$ and the profit at each stage is
realized by querying a noisy oracle. 
In many settings, $\X$ is a simple domain; here 
we assume that $\X\defeq [0,1]^d$ so that the inputs $\x$ are normalized.

The main challenges of the problem \eqref{eq:prob} that we seek to address are (i) the absence of structure (e.g., convexity) in $f$ that
could potentially be exploited, 
(ii) the high cost
associated with each
observation of $f$, and 
(iii) the time-varying setting where observations of $f$ cannot be made in the past or future.
Bayesian optimization (BO)~\cite{brochu2010tutorial, shahriari2015taking} with
Gaussian process (GP) priors~\cite{gramacy2020surrogates,rasmussen:williams:2006}  suits the
specific
challenges posed by our problem (mainly (i) and (ii) mentioned
above),
where observations at each round $i$, 
$\hat{y}_i =
f(\x_i, t_i) + \epsilon_i$, 
are made judiciously by leveraging information
gained from previous observations, as a means of coping with the high costs of
each observation. The key idea  is to specify GP prior distributions on the
oracle value and the noise. That is, $f(\x,t) \sim \mcl{GP}\left(0, k( (\x, t),
(\x',t'))\right)$ and $\epsilon \sim \mcl{GP}(0,\sigma^2_\epsilon)$, where $\sigma^2_\epsilon$ is a constant noise variance. 
The covariance function $k$ captures the correlation between the observations
in the joint $(\x,t)$ space; here we use the product composite form
given by $k((\x, t), (\x', t')) = k_{\x}(\x,\x'; \theta_\x) \times k_t(t,t';
\theta_t) $, where $\theta_\x$ and $\theta_t$ 
parameterize the covariance functions for $\x$ and $t$, respectively. We estimate the
GP hyperparameters 
$\bm{\Omega}= \lbrace \theta_x, \theta_t, \sigma^2_\epsilon \rbrace$
from data by maximizing the marginal
likelihood. We use the squared-exponential  kernels defined as $k_\x(\x, \x') = \T{exp} \left( -\f{\|\x -
\x'\|^2}{2 \theta^2_\x}\right)$ and $k_t(t, t') = \T{exp} \left( -\f{\|t -
t'\|^2}{2 \theta^2_t}\right)$.
The
posterior predictive distribution of the output $Y$, conditioned on available observations from the oracle, is given by 
\begin{equation}
    \begin{split}
        Y(\x, t |
      \D_n, \bm{\Omega}) &\sim \mcl{GP}(\mu_n(\x, t), \sigma^2_n (\x, t)),~ \quad \D_n \defeq  \lbrace
(\x_i, t_i), \hat{y}_i \rbrace _{i=1}^n \\
\mu_n(\x,t) &= \mbf{k}_n^\top [\mbf{K}_n + \sigma_\epsilon^2 \mbf{I}]^{-1} \mbf{y}_n\\
\sigma^2_n(\x, t) &=  k((\x,t), (\x,t)) - \mbf{k}_n^\top [\mbf{K}_n + \sigma_\epsilon^2 \mbf{I}]^{-1} \mbf{k}_n,
    \end{split}
    \label{e:GP}
\end{equation}
where $\mbf{k}_n$ is a vector of covariances between $(\x,t)$ and all observed
points in $\D_n$, $\mbf{K}_n$ is a sample covariance matrix of observed points in
$\D_n$, $\mbf{I}$ is the identity matrix, and $\mbf{y}_n$ is the vector of all
observations in $\D_n$; \eqref{e:GP} is then used as a surrogate model for
$f$ in Bayesian optimization. Note that $\mu_n$ and $\sigma_n^2$
are the posterior mean and variance of the GP, respectively, where the
subscript $n$ implies the conditioning based on $n$ past observations. BO then
proceeds by defining an \emph{acquisition function}, $\alpha(\x,t)$, in terms of the GP posterior that is
optimized in lieu of the expensive $f$ to select the next
point $\x_{n+1}$ at $t_{n+1}$, and the process continues until a budget of $q$ observations is reached; see Algorithm~\ref{a:BO}. We denote by $\x_i$ and $\x_t$ a decision made by maximizing the acquisition function at $t_i$ and $t$, respectively. We use $\x_t^*$ to denote a maximizer of $f(\cdot, t)$.

\begin{algorithm2e}[t]
\textbf{Given:} $\D_n = \lbrace
(\x_i, t_i), \hat{y}_i \rbrace _{i=1}^n$, 
 total budget $q$, 
 schedule $\lbrace t_{n+1},\ldots, t_{q}= T \rbrace$, 
 and GP hyperparameters $\bm{\Omega}$ \\
\KwResult{$(\x_T, \hat{y}_T)$}
  \For{$i=n+1, \ldots, q$, }{
  Find $\x_i \in \underset{\x \in \X}{\argmax}~ \alpha(\x, t_i)$ \qquad (acquisition function maximization based on $\bm{\Omega}$)\\
    Observe $\hat{y}_i$ = $f(\x_i, t_i) + \epsilon_i$\\ 
    Append $\D_i = \D_{i-1} \cup \lbrace (\x_i, t_i), \hat{y}_i \rbrace$\\ 
    Update GP hyperparameters $\bm{\Omega}$ \\
 }
 \caption{Generic Bayesian optimization}
 \label{a:BO}
\end{algorithm2e}

Typical acquisition functions in BO take a \emph{greedy} (``myopic'')
approach, where each $\x_i$ 
is selected as though it were the
last, without accounting for the potential impact on future selections.
Such myopic acquisition functions include the
probability of improvement (\texttt{PI})~\cite{kushner1964new}, the expected improvement
(\texttt{EI}) \cite{jones1998efficient, mockus1978application}, and the GP upper
confidence bound (\texttt{UCB})~\cite{srinivas2009gaussian}, which are, respectively,
\begin{equation}
    \begin{split}
        \alpha_{\PI}(\x, t) &= \Phi \left( z \right) \\
        \alpha_{\EI}(\x, t) &= z\sigma_n(\x, t) \Phi\left(z \right) + \sigma_n(\x, t) \phi \left( z \right) \\
        \alpha_{\UCB}(\x,t) &= \mu_n(\x, t) + \beta^{1/2} \sigma_n(\x, t),
    \end{split}
    \label{e:myopic_acqs}
\end{equation}
where $z = \f{\mu_n(\x, t) - \xi}{\sigma_n(\x, t)}$; $\Phi$ and $\phi$ denote
the standard normal cumulative and probability density functions, respectively;
$\xi$ is a user-specified target; and $\beta$ is a confidence parameter.
Whereas the \texttt{PI} and \texttt{EI} acquisition functions are probabilistic measures of
improvement over a user-specified target, \texttt{UCB} is an optimistic
estimate. 
Under suitable regularity conditions, the \texttt{PI}, \texttt{EI}, and
\texttt{UCB} acquisition functions are guaranteed to achieve asymptotic
consistency~\cite{bull2011convergence, mockus1978application,
srinivas2009gaussian, vazquez2010convergence} in BO. For a limited budget $q$, however, 
such acquisition functions can be suboptimal~\cite{gonzalez2016glasses}.
Furthermore, the myopic nature of these basic acquisition functions is such that they seek to maximize the 
oracle value at $t=T$ only when they arrive at $T$, meaning that the past ($t<T$) decisions 
are not necessarily made to facilitate making the best decision for 
$t=T$. On the other hand, we are interested in a strategy that seeks to
maximize $f(\x,T)$. 
Since the time horizon $t=T<\infty$, this is a \emph{finite-horizon} problem.

For time-dependent objectives, an optimal decision may be horizon-dependent,
which poses a challenge for improvement-based acquisition functions such as \texttt{PI} and \texttt{EI}. Because the target $\xi$ is typically chosen as the best observed value in $\{\hat{y}_i\}_{i=1}^n$, specifying an appropriate target when the maximum oracle value $\max f(\x,t)$
changes with time $t$ is challenging. For example, what if the maximum oracle value at $T$ is
much smaller than anything observed until current time $t$? A further challenge in our context is that observations at the target time $T$ are not realizable until we conclude the optimization. Therefore, in such situations, myopic 
acquisition functions that seek
the best oracle value at the current $t$ may not necessarily be  useful. Instead,
we are interested in acquisition functions that  optimize the
\emph{long-sighted} decision (at $t=T$) even if it means incurring a suboptimal
oracle value at intermediate $t<T$. 
In the BO literature, strategies that make decisions by accounting for the
decisions remaining in the future are generally referred to as \emph{lookahead}
approaches, which we adopt to solve our  finite-horizon problem. 

We now illustrate the benefit of a lookahead acquisition function in the
context of time-dependent BO. \Cref{fig:EI_v_2LEY} shows the difference
between the decisions made via EI and the lookahead acquisition function
\texttt{2LEY} introduced in \cref{s:mLBO}. The
top left of the figure shows contours of the true (noise-free) oracle with a target horizon of 
$T=3.0$. We assume we have made $n$ decisions before $t_{n+1}=2.625$, and we
are interested in making a decision at current $t_{n+1}$ (denoted by the horizontal
dashed line). The top right shows the $\texttt{2LEY}$ and $\texttt{EI}$
acquisition functions and the corresponding decisions made by maximizing each
of them (denoted by vertical dashed lines). In the bottom row, we show the GP
predictions at $T$, with each of the decisions made via the myopic and
lookahead acquisition functions incorporated, along with the noise-free $f(\x,
T)$ for reference, and the final decisions at $T$ (denoted by vertical dashed
lines). Notice that the lookahead decision results in a more accurate
prediction of the global maximum at $T$ compared with the myopic decision. This illustration shows that 
lookahead acquisition functions have the potential to make better decisions
than do myopic acquisition functions when the oracle value at a target (future) horizon is of interest. One could ask whether a randomly made decision that offers a \emph{space filling} effect (e.g., $0\leq \x \leq 0.4$ at $t=2.625$ ) could better learn $f(\x, T)$ and hence lead to a better final decision at $T$. We show that this is not the case, with a randomly chosen decision $(0.2474, 2.625)$ that does not outperform the lookahead method; see bottom right of \Cref{fig:EI_v_2LEY}.

\begin{figure}[htbp!]
    \centering
    \begin{subfigure}{1\textwidth}
        \centering
        \includegraphics[width=1\linewidth]{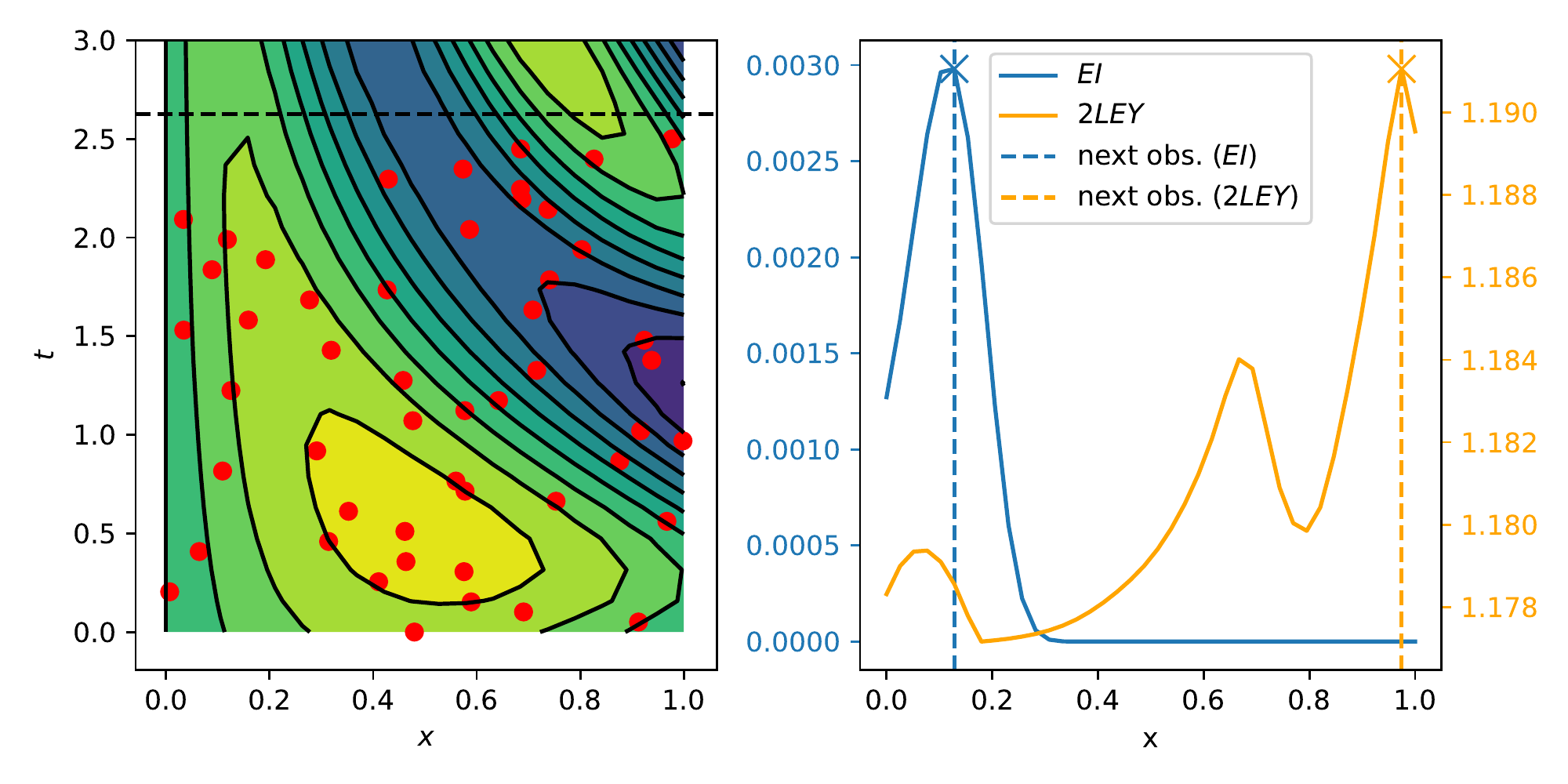} 
        \caption{Left: contour of $f(\x, t)$ in $\X\times \mcl{T}$ with $T=3$. Red circles denote the locations of $n$ previous observations. Horizontal line denotes the time $t_{n+1}=2.625$ at which the next observation is to be made. Right: myopic (\texttt{EI}) and lookahead (\texttt{2LEY}) acquisition functions. }
    \end{subfigure}\\
    \begin{subfigure}{1\textwidth}
        \centering
        \includegraphics[width=1\linewidth]{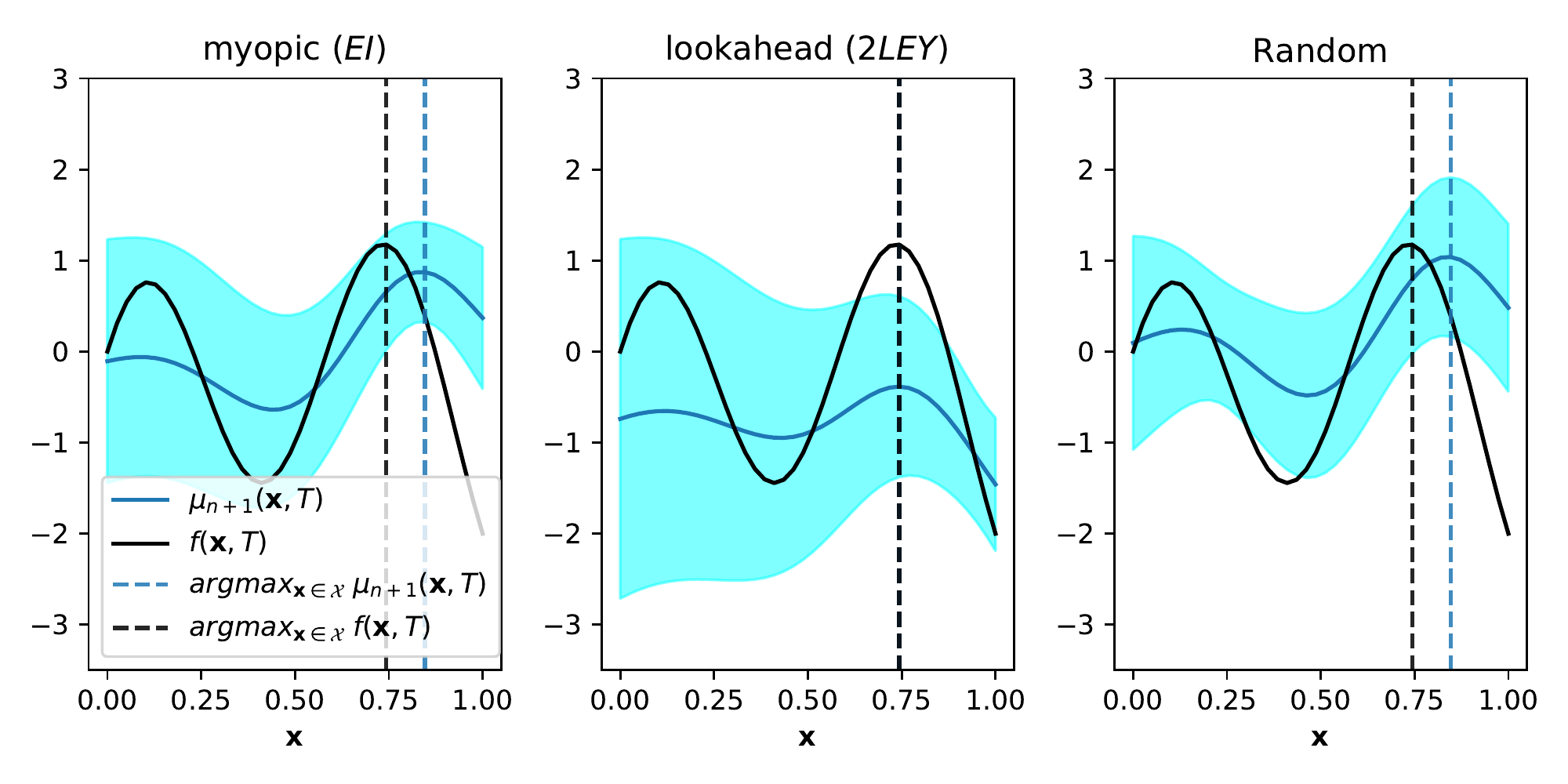}   
        \caption{GP prediction at $T$ with decisions incorporated from myopic (\texttt{EI}) and lookahead (\texttt{2LEY}) decisions as well as a decision chosen uniformly at random from $\bs{\mcl{X}}$ at $t$. Shaded region represents 95\% confidence region.}
    \end{subfigure}
    \caption{Comparison of the decisions made via a myopic (\texttt{EI}) 
    and a lookahead (\texttt{2LEY}) acquisition function. The top right shows the acquisition functions and the respective decisions (maximizers). The GP predictions, incorporating the decision made via either of the acquisition functions, are shown in the bottom. One can see that if the final decision at $T$ is made as   $\argmax_{\x\in\X}~\mu_{n+1}(\x, T)$, then the lookahead approach predicts the maximizer at $T$ more accurately than does the myopic approach. A random decision $(0.2474, 2.625)$, even though it is \emph{space filling}, does not outperform the lookahead decision.}
    \label{fig:EI_v_2LEY}
\end{figure}

In foundational work on lookahead BO, Osborne \cite{osborne2010bayesian}
showed that proper Bayesian reasoning can be used to define an acquisition
function where the ($n+1$)th observation is made by marginalizing a loss function
at the final ($q$th) observation  over all the remaining observations $(\x_i, \hat{y}_i)$,
$i=n+2,\ldots, q$. Others have proposed
\emph{lookahead EI} in the context of finite-budget BO. For example, Ginsbourger and Le
Riche \cite{ginsbourger2010towards} showed that the lookahead EI is a dynamic
program that chooses the EI maximizer in expectation, considering all possible
strategies of the same budget. 
A practical challenge with lookahead strategies is the computational
tractability (as demonstrated in, for example, \cite{ginsbourger2010towards,osborne2010bayesian}), which can be partially mitigated by various
approximation methods \cite{gonzalez2016glasses,lam2016bayesian,wu2019practical}. Yue and Kontar \cite{yue2020non} explored the
theoretical properties of the approximate dynamic programming approach for
multistep lookahead acquisition functions and discussed when these strategies are guaranteed
to perform better than their myopic counterparts. In this work
we draw inspiration from \cite{osborne2010bayesian} to address the
limited budget aspect of our problem. However, we propose a general framework
for lookahead BO that readily applies to the finite-horizon,
time-dependent BO; we are unaware of any previous work on this topic.

In the context of
time-dependent BO, the only works we are aware of are
\cite{nyikosabayesian} and \cite{bogunovic2016time}. 
We fundamentally differ
from their approach in the following ways: (i) their goal is to track a
time-varying optimum whereas ours is to predict it at the target horizon $T$), and (ii) they do not address
the limited budget constraint as we do. Furthermore, setting $k_t(t_i,
t_j)=(1-\epsilon)^{|t_i - t_j|/2}$ (where $\epsilon$ is the \emph{forgetting}
factor as defined in \cite{bogunovic2016time}), our method is equivalent to
\cite{bogunovic2016time} if we myopically make decisions at each $t$ via \texttt{UCB}.


Another plausible way to solve the finite-horizon (target $t=T$) BO problem
is to
sequentially choose points that are most
\emph{informative}~\cite{cover2012elements} about the maximizer
at the target horizon; see, for example, 
\cite{hennig2012entropy, hernandez2014predictive, wang2017max}. Although such
information-theoretic approaches are known to involve an intractable form of
the information entropies, they are also highly reliant on an accurate underlying GP model~\cite{mackay1992information}.

In this work, we extend myopic acquisition functions in order to make
nonmyopic decisions for
finite-horizon BO. Specifically, we propose a generalized
framework that specifies a \emph{value function} 
that quantifies an objective at the target horizon $T$ and that can then be used in a lookahead acquisition
function. Our main contributions are as follows. 
\begin{enumerate}
    \item We present the first strategy to solve finite-horizon BO for time-dependent 
    oracles.
    \item We propose lookahead versions of myopic acquisition functions~(\cref{s:mLBO}) that
      are applicable to both finite-budget BO and finite-horizon, time-dependent settings. 
    \item We present a practical algorithm that recursively applies a two-step
      lookahead acquisition function~(\cref{s:r2L}) 
      to make new observations time-dependent oracles.
    \item We establish the utility of our method by demonstrating it on carefully constructed synthetic test functions (\cref{ss:synthetic_1d,ss:synthetic_hd}) and a real-world quantum optimal control problem~(\cref{ss:qoc}). 
\end{enumerate}
Software that implements our lookahead framework for limited horizon, finite-budget time-dependent BO problems is available upon request and will be released publicly with the final version of this manuscript.

The remainder of the paper is organized as follows. In  \cref{s:mLBO}
we present our generalized lookahead acquisition framework for generic value
functions. In \cref{s:r2L} we present a recursive two-step lookahead
approach that is more practical than a more-than-two-step lookahead approach.
We discuss theoretical properties of our approach in \cref{s:theory}.
We then
present the results of our numerical experiments on synthetic test functions
(\cref{ss:synthetic_1d,ss:synthetic_hd}) as well as a
real-world quantum optimal control problem in \cref{ss:qoc}. We conclude the paper with an outlook on future work.

\section{Generalized m-step lookahead acquisition function}
\label{s:mLBO}
Given the data from $n$ observations $\D_n$, let us define a value function $\upsilon$ by 
\begin{equation}
    \upsilon_n(\x, t; \bs{\omega}) \defeq \mbb{E}_{y\sim Y|\D_n}\left[
    h(y(\x, t); \bs{\omega})     
    \right],
    \label{e:value_function}
\end{equation}
where $h$ is a scalar-valued function, $\bs{\omega} \in \mbb{R}^p$ are
parameters that parameterize $h$, 
$Y|\D_n$ is the GP posterior 
given $n$ observations, and $y(\x, t)$ is a sample path realized from $Y(\x,
t)$. For example, the \texttt{PI}, \texttt{EI}, and \texttt{UCB} acquisition functions can be obtained from defining $h$ as 

\begin{equation}
    \begin{split}
    h(y(\x, t); \xi) 
    &= y(\x, t) - \xi \implies \upsilon(\x,t;\xi) = \Phi \left( z \right) \qquad \text{(PI)}\\
    h(y(\x, t); \xi) 
    &= [y(\x, t) - \xi]^+ \implies \upsilon(\x,t;\xi) = z\sigma(\x, t) \Phi\left(z \right) + \sigma_n(\x, t) \phi \left( z \right) \qquad \text{(EI)}\\
        h(y(\x, t); \beta)
        &= y(\x, t) +  \beta^{1/2}  \sigma(\x,t) \implies \upsilon(\x,t;\beta) = \mu(\x, t) + \beta^{1/2}\sigma(\x, t), \qquad \text{(UCB)} 
        \end{split}
    \label{e:value_functions}
\end{equation}
where recall that $\xi$ and $\beta$ are user-specified target and confidence parameters, respectively.
Similarly, when $h$ is the identity function $h(y(\x, t); \bs{\omega})  = y(\x, t)$, then $\upsilon(\x,t)= \mu(\x, t)$. 

We use the following notation for conciseness:
\begin{align*}
&\x^t_{n+1} \defeq (\x, t_{n+1}) && \D_{n,j} \defeq \D_n \bigcup_{i=1,\ldots,j} \{ \x^t_{n+i}, y_{n+i} \}\\ 
&y_{n+1} \defeq y(\x^t_{n+1}) &&\mbb{E}_{n,j} \defeq \mbb{E}_{y\sim Y|\D_{n,j}} \\
&\mu_{n,j} \defeq \mbb{E}_{n,j}Y, ~~\sigma^2_{n,j} \defeq \mbb{E}_{n,j} \left( Y - \mu_{n,j}\right)^2 && \upsilon^T_{n,j} \defeq \upsilon_{n,j}(\x, T; \bs{\omega}) = \mbb{E}_{n,j}~\left[h(y(\x, T); \bs{\omega}) \right],
\end{align*}
where we note
that whenever we explicitly use $y(\x,t)$, we refer to a sample path drawn from
$Y(\x,t)$ and therefore $y_{n+1} = y(\x_{n+1},t_{n+1})$ is a draw from $y(\x,t)$ at $(\x_{n+1}, t_{n+1})$. 
On the other hand, when we use $\hat{y}_i$, we refer to an observation of the oracle at
$(\x_i,t_i)$. Similarly, note that $\D_{n,k}$ contains $n$ oracle evaluations 
and $k$ candidate points $\x$ and the $k$ corresponding draws from the GP at times 
$t_{n+1}$, \ldots, $t_{n+k}$. Furthermore, we note here that $\x^t_{n+1}$ is $d$-dimensional.

We define our $m$-step lookahead acquisition function as
\begin{equation}
    \begin{split}
      \alpha_m(\x^t_{n+1}) = \int \ldots \int  & \left[\max_{\tilde{\x} \in \X}~\upsilona{n}{m-1}(\tilde{\x}, T; \bs{\omega})|\D_{n,m-1}\right] ~ p(y_{n+m-2}|\D_{n,m-2})\\ &p(\x^t_{n+m-1}|\D_{n,m-2})\cdots 
    p(y_{n}|\D_n)~dy_{n+m-2}\cdots dy_{n}~d\x^t_{n,m-1} \cdots d\x^t_{n+2}.
    \end{split}
    \label{e:mstep}
\end{equation}
Note that in \eqref{e:mstep}, the \emph{inner maximization} is always with respect to $\x$ at $t=T$ and that after evaluating the integral, the resulting expression is a function of $\x^t_{n+1}$. Furthermore, since we are interested in choosing the maximizer of $\alpha_m(\x^t_{n+1})$ for $t_{n+1}$, we do not marginalize $\x^t_{n+1}$. Computing \eqref{e:mstep} involves the computation of $m-1$ nested expectations. Additionally, when the expectations are approximated via Monte Carlo sampling, there is a total of $N\times (m-1)$ $d$-dimensional  optimization solves due to the inner maximization in \eqref{e:mstep}. The computation or even the Monte Carlo approximation of \eqref{e:mstep} becomes intractable as $m$ increases, warranting further approximation, as is done, for instance, in \cite{gonzalez2016glasses, lam2016bayesian, yue2020non}. When $m=2$, however, the computation of \eqref{e:mstep} is much more straightforward, as we discuss in the following section.

\section{Generalized two-step lookahead acquisition function}
\label{s:r2L}
In this work, as a practical alternative to the computationally intractable multistep ($m>2$) 
lookahead acquisition function, we recursively optimize the two-step lookahead acquisition function~\eqref{e:E_twostep} at each time step of the time-dependent BO problem; this process is graphically described in \cref{f:method}.
Given $\D_n$ and with $m=2$, we define our \emph{two-step} lookahead acquisition function by
\begin{equation}
\begin{array}{rl}
  \alpha_2(\x,t_{n+1})
  
&=  \displaystyle \int \max_{\tilde{\x} \in \X}~ \left[ \upsilon_{n,1}(\tilde{\x}, T; \bs{\omega}) | \, 
\D_n \bigcup \{ (\x,t_{n+1}), y_{n+1} \}
\right]~ p(y_{n}|\D_{n}) ~dy_{n}, 
\end{array}
\label{e:twostep}
\end{equation}
where $p(y_{n}|\D_{n})$ is the probability density of the GP posterior distribution
conditioned on $\D_n$.
We write \eqref{e:twostep} concisely as 
\begin{equation}
  \alpha_2(\x^t_{n+1}) = \mbb{E}_n \left[\max_{\tilde{\x} \in \X}~\upsilona{n}{1}(\tilde{\x},T;\bs{\omega})|\D_{n,1} \right].
    \label{e:E_twostep}
\end{equation}
In \eqref{e:E_twostep}, recall that $\D_{n,1}$ is the union of $\D_n$ and the $n+1$th observation being a draw from the GP $Y_n$, and we compute $\upsilon_{n,1}(\tilde{\x}, T; \bs{\omega})$ by taking the expectation with respect to $Y|\D_{n,1}$. In general, it is not guaranteed that
\eqref{e:E_twostep} admits a closed-form expression. Hence we make the
Monte Carlo approximation
\begin{equation}
    \alpha_2(\x^t_{n+1}) \approx \hat{\alpha}_2(\x^t_{n+1}) \defeq \frac{1}{N} \sum_{j=1}^N \max_{\tilde{\x} \in \X}~\upsilona{n}{1}^{T,j},
    \label{e:twostep_mc}
\end{equation}
where $\upsilona{n}{1}^{T,j} = \upsilona{n}{1}(\tilde{\x}, T)|y^j_{n+1}$ with the $N$ iid samples $\{y^j_{n+1}: \,j=1,\ldots,N\}$ drawn from $Y_{n}$. In practice, the computation of \eqref{e:E_twostep} involves the following steps (given $\x_{n+1}$):
\begin{enumerate}
    \item Generate $N$ iid samples $\{ y^j_{n+1}\}_{i=1}^N$ drawing from $\mcl{N}(\mu_n(\x_{n+1}, t_{n+1}), \sigma_n^2(\x_{n+1}, t_{n+1}))$.
    \item For each $y^j_{n+1},~j=1,\ldots,N$
    \begin{enumerate}
        \item[2.1] update $\D_n$ to $\D_{n,1}$ to get GP posterior $Y|\D_{n,1}$ and
        \item[2.2] maximize a draw of the projected value function $\nu^j = \max_{\tilde{\x} \in \X}~\upsilona{n}{1}(\tilde{\x}, T)|\D_{n,1}$.
    \end{enumerate}
    \item Compute $\hat{\alpha}_2(\x_{n+1}) \approx \frac{1}{N} \sum_{j=1}^N \nu^j$.
\end{enumerate}

In order to make a decision at $t_{n+1}$, \eqref{e:E_twostep} needs to be maximized. We show how the maximization of the acquisition function can be done efficiently by constructing unbiased estimators of the acquisition function's gradient. Under suitable regularity conditions, we can interchange the expectation and gradient
operators (see \cref{thm:gradexp}
) to obtain
\begin{equation}
    \nabla_{\x}\alpha_{2}(\x^t_{n+1}) = \mbb{E}_{n} \left[ \nabla_{\x}~ \max_{\tilde{\x} \in \X}~\upsilona{n}{1}(\tilde{\x}, T) | \D_{n,1} \right].
    \label{e:twostep_T_grad}
\end{equation}
Similar to the Monte Carlo approximation for the acquisition function presented
in \eqref{e:twostep_mc}, \eqref{e:twostep_T_grad} can be approximated as
\begin{equation}
    \nabla_{\x}\alpha_{2}(\x^t_{n+1}) \approx \tilde{\nabla}_{\x}\alpha_{2}(\x^t_{n+1}) = \f{1}{N} \sum_{i=1}^{N} d\nu^i \defeq  \f{1}{N} \sum_{i=1}^{N}  \nabla_{\x}~ \max_{\tilde{\x} \in \X}~\upsilona{n}{1}^{T,i}.
    \label{e:twostep_T_grad_mc}
\end{equation}
In this way, sample-averaged gradients can be used in a gradient-based optimizer
to efficiently optimize \eqref{e:E_twostep}. We clarify the issues associated
with computing the gradients of functions involving the $\max()$ operator as in
\eqref{e:twostep_T_grad_mc} in \Cref{s:theory}. 

We now illustrate in detail the two-step
acquisition function for a specific value function.
\subsection*{Two-step lookahead expected oracle value 
(\texttt{2LEY}) acquisition function}
\label{ss:r2LEY}
\begin{figure}
    \centering
    \includegraphics[width=\textwidth, trim=0.5cm 4cm 0.5cm 6cm,clip]{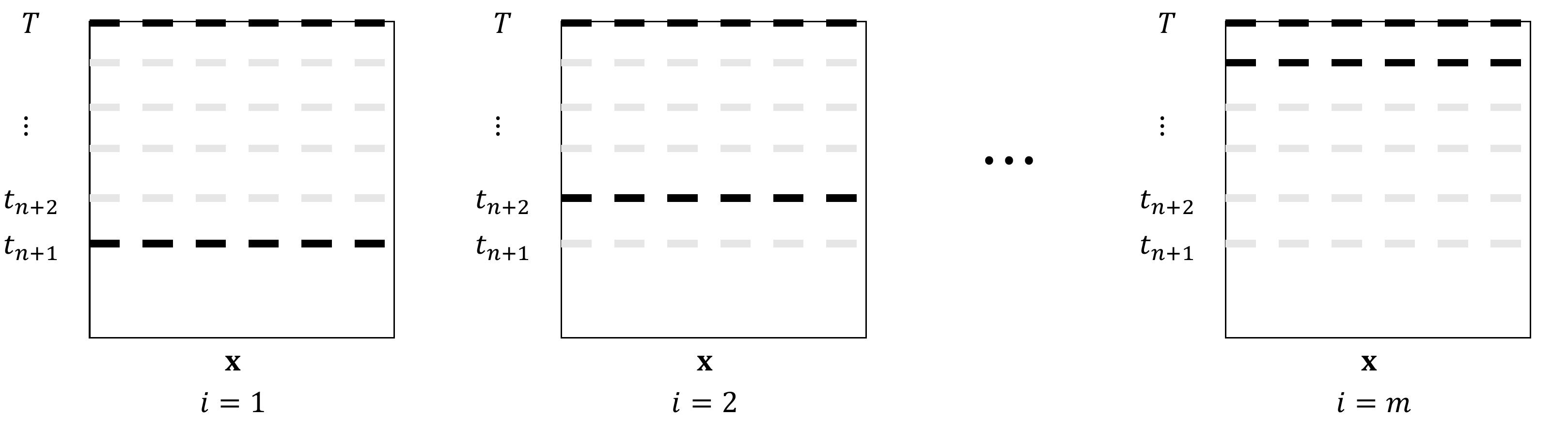}
    \caption{Recursive two-step lookahead method for time-dependent problems.
    Here, we start with $n$ observations up to $t_n$ and have $q-n$ 
    remaining
    observations to make, where $q$ is the total budget. We sequentially select
    one point $\x_{n+i}$ at each step that maximizes the expected value
    function evaluated at $T$; this is denoted by the black dashed lines
    at $t_{n+i}$ and $T$, whereas the gray dashed lines denote the
    remaining decisions to be made. We repeat the same process until the
    decision is made at the penultimate step $i=q-1$. The final decision at $T$
    is made by maximizing the value function at $T$.}
    \label{f:method}
\end{figure}
As an illustration, we set $h$ to be the identity function, and hence $\upsilona{n}{1}(\x, T) = \mu_{n,1}(\x,
T)|\D_{n,1}$.
We define our two-step lookahead expected oracle value (\texttt{2LEY}) acquisition function as
\begin{equation}
    \begin{split}
    \alpha_{2LEY}(\x^t_{n+1}) \defeq & \mbb{E}_n \left[  \underset{\tilde{\x} \in \X}{\max}~\left.\mu_{n,1}(\tilde{\x}, T) \right| \D_{n,1} \right] \\
        =& \mbb{E}_n \left[ \mu^*_{n,1}( \x^t_{n+1}, T) \right],
    \end{split}
    \label{e:2LEY_1}
\end{equation}
where $\mu^*_{n,1} = \mu_{n,1}(\x^*, T)$ with $\x^*$ being a maximizer of $\mu_{n,1}(\x,
T)|\D_{n,1}$ and 
$y_{n+1}$ is a draw from $\mcl{N} \left(\mu_n(\x^t_{n+1}), \sigma_n^2(\x^t_{n+1}) \right)$. 
The dependence of the second line of \eqref{e:2LEY_1} on
$\x^t_{n+1}$ can
be seen by realizing that
\begin{equation}
    \mu_{n,1}^*(\x^t_{n+1}, T)  = \mbf{k}_{n+1}^\top \mbf{K}^{-1}_{n+1} [\mbf{y}_n^\top, y_{n+1}]^\top,
\end{equation}
where $ \mbf{k}_{n+1} = [ k\left((\x^*,T), (\x_1, t_1)\right), \ldots,
k\left((\x^*,T), (\x_n, t_n)\right), k\left((\x^*,T),
(\x_{n+1},t_{n+1})\right)]^\top$, and $\mbf{K}_{n+1, ij} = k\left((\x_i, t_i), (\x_j, t_j) \right);~\forall i=1,\ldots,n+1, ~ j=1,\ldots,n+1$. The expectation in \eqref{e:2LEY_1} is approximated as
\begin{equation}
    \alpha_{2LEY}(\x^t_{n+1}) = \mbb{E}_n \left[ \mu_{n,1}^*(\x^t_{n+1}, T) \right] \approx \frac{1}{N} \sum_{j=1}^N  \mu_{n,1}^*(\x^t_{n+1}, T) | y^j_{n+1}.
    \label{e:2LEY_MC}
\end{equation}
The gradient of $\mu_{n,1}^*(\x^t_{n+1}, T)$ with respect to $\x^t_{n+1}$ is given by
\begin{equation}
    \begin{split}
        \nabla_{\x} \mu_{n,1}^*(\x^t_{n+1}, T) = &\f{\partial \mbf{k}_{n+1}^\top}{\partial \x_{n+1}} \mbf{K}^{-1}_{n+1} + \mbf{k}_{n+1}^\top \f{\partial \mbf{K}^{-1}_{n+1}}{\partial \x_{n+1}} \\
        = &\f{\partial \mbf{k}_{n+1}^\top}{\partial \x_{n+1}} \mbf{K}^{-1}_{n+1} +
        \mbf{k}_{n+1}^\top~\mbf{K}^{-1}_{n+1} \f{\partial \mbf{K}_{n+1}}{\partial \x_{n+1}}\mbf{K}^{-1}_{n+1},
    \end{split}
    \label{e:2LEY-grad_a}
\end{equation}
where $ \f{\partial \mbf{K}_{n+1}}{\partial \x_{n+1}}$ is a matrix of
elementwise derivatives and the second line follows from a well-known lemma on the derivative of a matrix inverse~\cite[pp.~201--202]{rasmussen:williams:2006}.

With a continuously differentiable kernel $k(\cdot, \cdot)$, we have that
\begin{equation}
    \nabla_{\x} \alpha_{2LEY}(\x^t_{n+1}) = \nabla_{\x} \mbb{E}_n \left[ \mu_{n,1}^*(\x^t_{n+1}, T) \right] = \mbb{E}_n \left[\nabla_{\x} \mu_{n,1}^*(\x^t_{n+1}, T) \right],
    \label{e:2LEY_gr_ex_interchange}
\end{equation}
where the interchange of the gradient and expectation operators is via \cref{thm:gradexp}. Then the gradient is approximated as 
\begin{equation}
    \nabla_{\x} \alpha_{2LEY}(\x^t_{n+1}) = \mbb{E}_n \left[\nabla_{\x} \mu_{n,1}^*(\x^t_{n+1}, T) \right] \approx \frac{1}{N} \sum_{j=1}^N \nabla_{\x} \mu_{n,1}^*(\x^t_{n+1}, T) | y^j_{n+1}.
    \label{e:2LEY_grad_MC}
\end{equation}
Furthermore, \eqref{e:2LEY_MC} and
\eqref{e:2LEY_grad_MC} are unbiased estimators for $\alpha_{2LEY}$ and $\nabla_{\x}
\alpha_{2LEY}$, respectively; see \cref{thm:gradexp}. This gradient can then be used in a gradient-based optimizer to
maximize our two-step lookahead acquisition function in an efficient manner.

Similarly, the two-step lookahead  extensions for
\texttt{EI}, \texttt{PI}, and \texttt{UCB} are obtained by picking the
corresponding value functions as $\alpha_{\EI}(\x,T)$, $\alpha_{\PI}(\x,T)$ and $\alpha_{\UCB}(\x,T)$ and are
named respectively, $\alpha_{2LEI}$, $\alpha_{2LPI}$, and $\alpha_{2LUCB}$. Our
overall algorithm recursively maximizes the two-step acquisition function until
the penultimate step and maximizes the actual value function at the final
step. The overall algorithm is graphically depicted in \cref{f:method} and is presented in
Algorithm~\ref{a:r2LBO}. We refer to Algorithm~\ref{a:r2LBO} with the value functions chosen as $\mu_{n,1}(\x, T)$, $\alpha_{EI}$, $\alpha_{PI}$, and $\alpha_{UCB}$ as \texttt{r2LEY}, \texttt{r2LEI}, \texttt{r2LPI}, and \texttt{r2LUCB}, respectively.


\begin{algorithm2e}[H]
\SetAlgoLined
\textbf{Given:} $\D_n$,  
total budget $q$, 
schedule $\lbrace t_{n+1},\ldots, t_{q}\defeq  T \rbrace$, and GP hyperparameters $\bm{\Omega}$ \\
\KwResult{$(\x_T, \hat{y}_T)$}
  \For{$i=n+1, \ldots, q-1$, }{
  Define value function $\upsilon^T_{i-1, 1}|\D_{i-1, 1}$ \\
  Find $\x_i = \underset{\x \in \X}{\argmax}~ \hat{\alpha}_2(\x, t_i)$, \qquad (acquisition function maximization)\\
  \quad with $\hat{\alpha}_2$ and $\nabla_{\x}~\hat{\alpha}_2$ evaluated via \eqref{e:twostep_mc} and \eqref{e:twostep_T_grad_mc} resp. \\
    Observe $\hat{y}_i$ = $f(\x_i, t_i) + \epsilon_i$\\ 
    Append $\D_i = \D_{i-1} \cup \lbrace (\x_i, t_i), \hat{y}_i \rbrace$\\ 
    Update GP hyperparameters $\bm{\Omega}$ \\
 }
 Find $\x_T \in \underset{\x \in \X}{\argmax}~\upsilon_{q}(\x, T)|\D_{q-1}$ \qquad (value function maximization)\\
 Observe $\hat{y}_T$ = $f(\x_T, T) + \epsilon_T$\\ 
 \caption{Recursive two-step lookahead Bayesian optimization} 
 \label{a:r2LBO}
\end{algorithm2e}

\section{Theoretical properties}
\label{s:theory}

We now present the theoretical foundation of our method and a few properties.  
  Let $\mcl{Y}$
  be the support of $y_{n+1}$, whose density is $p(y_{n+1})$. As before, let the
  domain of $\x^t_{n+1}$ be $\X$. 
We make the following remark about the differentiability of our acquisition function; this justifies the interchange of the gradient and expectation operators.

\begin{remark}
We write the inner term of our two-step lookahead acquisition function in \eqref{e:E_twostep} as $\max_{\x\in\X} \mbb{E}_{n,1}h(y)$, where we omit the arguments of $y$ and $\bs{\omega}$ for the sake of conciseness. The GP posterior mean $\mu(\x, t)$ and standard deviation $\sigma(\x, t)$ are continuously differentiable with the sufficient condition that we choose a continuously differentiable prior covariance kernel $k$~\cite{smith1995differentiation, wang2016parallel}. We then use the reparameterization trick~\cite{wilson2017reparameterization,
wilson2018maximizing} to write the GP sample path as $y(\x, T) \defeq \mu(\x, T) + \sigma(\x,
T)\times \gamma$, where $\gamma \sim \mcl{N}(0, 1)$, which makes the differentiability of our acquisition function transparent. 


From a purely implementation standpoint, we use the gradient $\nabla
\max_{\tilde{\x}\in\X}{\mbb{E}_{n, 1}~h(y)}$ if it exists and a
\emph{subgradient}~\cite{boyd2004convex} otherwise. In our implementation, the
gradients are computed by backpropagating the output of the
$\max()$ operation; in \Cref{app:grad_nd} we provide examples of how (sub)gradients are calculated for nonsmooth functions.
\end{remark}

Let $\x^*$ be a maximizer of $\mbb{E}_{n,1} h(y(\x, T); \bs{\omega}) = \upsilon_{n,1}(\x, T)|\D_{n,1}$. In what follows, we write
  $g^*(\x^t_{n+1}, T, y_{n+1}) \defeq \upsilon_{n,1}(\x^*, T)|\D_{n,1}$ to denote the maximum value of the value function at $T$ after having updated $\D_n$ with $y_{n+1}$, and $\nabla_{\x}g^* \defeq \partial g^*(\x^t_{n+1}, T, y_{n+1})/\partial \x^t_{n+1} $. Note again that the dependence of $g^*$ on $\x^t_{n+1}$ is due to $\D_{n,1}$.
  
\begin{theorem}[Interchange of gradient and expectation operators] 
\label{thm:gradexp}
 Let $\nabla_{\x}g^*$ exist, and let $g^*(\x^t_{n+1}, T, y_{n+1})\times
  p(y_{n+1})$ and $\nabla_{\x}g^* \times
  p(y_{n+1})$ be continuous on $\X \times \mcl{Y}$. 
  Further assume
  that there exist functions $q_0$ and $q_1$ such that
\begin{equation*}
\begin{array}{rl}
  |g^*(\x^t_{n+1}, T, y_{n+1}) \times p(y_{n+1})| &\leq q_0(y_{n+1}) \\
  \|
  \nabla_{\x} g^* \times  p(y_{n+1})\| &\leq q_1(y_{n+1})
  \end{array}
  \qquad \forall (\x^t_{n+1}, y_{n+1}) \in \X \times \mcl{Y}, 
  \end{equation*}
where
  $\int_{\mcl{Y}}q_1(y_{n+1}) dy_{n+1} < \infty$ and
  $\int_{\mcl{Y}}q_2(y_{n+1}) dy_{n+1} < \infty$. Then,
\[ \nabla_{\x} \mbb{E}_j \left[ g^*(\x^t_{n+1}, T, y_{n+1}) \right] = \mbb{E}_n \left[\nabla_{\x} g^*(\x^t_{n+1}, T, y_{n+1}) \right].\]

Furthermore, a realization of $\nabla_{\x} g^*(\x^t_{n+1}, T, y_{n+1})$ yields
an unbiased estimate of the gradient of $\alpha_2(\x_{n+1}^t)$ in \eqref{e:E_twostep}.
\end{theorem}

\begin{proof}
Using $g^*(\x^t_{n+1})$ to denote  $g^*(\x^t_{n+1}, T, y_{n+1})$ for the sake of brevity, we have
\begin{equation}
\begin{split}
    \nabla_\x \alpha_2(\x^t_{n+1}) \defeq \frac{\partial}{\partial \x^t_{n+1}} \mbb{E}_n \left[ g^*(\x^t_{n+1}) \right] &= \underset{h\rightarrow 0}{\T{lim}}~\frac{1}{h} \left \lbrace \mbb{E}_n \left[ g^*(\x^t_{n+1}+h) \right] -  \mbb{E}_n \left[ g^*(\x^t_{n+1}) \right] \right \rbrace \\
    &= \underset{h\rightarrow 0}{\T{lim}}~ \left \lbrace \mbb{E}_n \frac{1}{h} \left[ g^*(\x^t_{n+1}+h) - g^*(\x^t_{n+1}) \right] \right \rbrace \\
    &= \underset{h\rightarrow 0}{\T{lim}}~ \left \lbrace \mbb{E}_n \frac{\partial g^*(\bar{\x}^h_{n+1})}{\partial \x_{n+1}}  \right \rbrace ,\\
\end{split} 
\label{e:grad_limit}
\end{equation}
where $\bar{\x}^h_{n+1} = \lambda \x^t_{n+1} + (1-\lambda)(\x^t_{n+1} + h)$ for some
$\lambda \in [0,1]$ and the last line follows from the first-order mean value
theorem.
We bring the limit inside the integral by Lebesgue's dominated convergence theorem to obtain
\begin{equation}
    \frac{\partial}{\partial \x^t_{n+1}} \mbb{E}_n \left[ g^*(\x^t_{n+1}) \right] =  \mbb{E}_n~\underset{h\rightarrow 0}{\T{lim}} \frac{\partial g^*(\bar{\x}^h_{n+1})}{\partial \x^t_{n+1}}   = \mbb{E}_n \frac{\partial g^*(\x^t_{n+1})}{\partial \x^t_{n+1}}.
    \label{e:LDC}
\end{equation}
Since \eqref{e:grad_limit} holds, a realization $\partial g^*(\x^t_{n+1}) / \partial \x^t_{n+1}$ yields an unbiased estimate of the gradient $\nabla_\x \alpha_2(\x^t_{n+1})$.
\end{proof}

\begin{proposition}[Special case when $m=2$ and $\upsilon^T = \mu(\x, T)$] 
When $m=2$, our \texttt{2LEY} acquisition function is equivalent to the knowledge gradient
(KG)~\cite{frazier2008knowledge, wu2016parallel} acquisition function offset by a constant factor. 

\begin{proof}
The KG acquisition function, in the context of our problem, is given as
\begin{equation}
    \alpha_{KG}(\x^t_{n+1}) = \mbb{E}_{n}\left[ \underset{\tilde{\x} \in \X}{\max}~ \mu_{n,1}(\tilde{\x},T)|\D_{n,1} \right] - \underset{\x \in \X}{\max}~ \mu_{n}(\x,T).
    \label{e:KG_time}
\end{equation}

With $m=2$, the \texttt{2LEY} acquisition function at $\x_{n+1}$ is defined as follows:
\begin{equation}\label{eq:nm1}
    \begin{split}
    \alpha_{2LEY}(\x^t_{n+1}) & = \int_{y_{n}} \left[ \max_{\tilde{\x} \in \X} ~\mbb{E}_{n,1} ~y(\tilde{\x},T)|\D_{n,1} \right] p(y_{n}|\D_{n})~dy_{n} \\
    & = \int_{y_{n}}\left[  \max_{\tilde{\x} \in \X}~ \mu_{n,1}(\tilde{\x},T)|\D_{n,1} \right] p(y_{n}|\D_{n})~dy_{n} \\
    & = \mbb{E}_{n}\left[ \max_{\tilde{\x} \in \X}~ \mu_{n,1}(\tilde{\x},T)|\D_{n,1} \right] \\
    & = \alpha_{KG}(\x^t_{n+1}) + \max_{\tilde{\x} \in \X}~\mu_{n}(\tilde{\x},T), 
    \end{split}
\end{equation}
where the last line follows from \cref{e:KG_time}. Note that $\alpha_{2LEY}$ and $\alpha_{KG}$ share the same maximizer. 
\end{proof}
\label{prop:special_cases}
\end{proposition}

\begin{proposition}[Special case when $|T-t| \gg\theta_t$] 
With the prior assumption $f(\x,t)\sim \mcl{GP}(0, k(\cdot,\cdot))$, where $k(\cdot,\cdot) \defeq k_{\x}(\cdot,\cdot; ~\theta_{\x})\times k_t(\cdot,\cdot; ~\theta_{t})$, with $k_{\x}$ and $k_t$ being squared-exponential kernels,
our two-step lookahead acquisition function approaches $\alpha_2(\x^t_{n+1}) = \max_{\tilde{\x}\in\X}{\upsilon_n(\tilde{\x}, T)}$ as $|T - t_{n+1}| \rightarrow \infty$. In other words, far away from the target horizon $T$, $\alpha_2(\x^t_{n+1})$ approaches a constant function.


\begin{proof}
When $|T-t_{n+1}| \gg \theta_t$, then $k_t(t_{n+1}, T) \approx 0$; this yields $k\left( (\x,t_{n+1}), (\x,T) \right) \approx 0$. Let $\mu_n$ and $\sigma_n^2$ be the posterior mean and variance given $\D_n$, respectively. Given another observation 
$y_{n+1}$ drawn from the GP $Y_n$ at $(\x, t_{n+1})$, where $|T-t_{n+1}| \gg \theta_t$, the updated posterior mean and variance are given by
\begin{equation}
    \begin{split}
        \mu_{n,1} =& \mbf{k}^\top_{n+1}\mbf{K}^{-1}_{n+1}[\mbf{y}^\top_n, y_{n+1}]^\top  \approx [\mbf{k}^\top_n, 0] \begin{bmatrix}
        \mbf{K}_n & \mbf{0} \\
        \mbf{0} & 1
        \end{bmatrix}^{-1}[\mbf{y}^\top_n, y_{n+1}]^\top = \mu_n \\
        \sigma^2_{n,1} =& 1-\mbf{k}^\top_{n+1}\mbf{K}^{-1}_{n+1}\mbf{k}_{n+1} \approx 1 - [\mbf{k}^\top_n, 0] \begin{bmatrix}
        \mbf{K}_n & \mbf{0} \\
        \mbf{0} & 1
        \end{bmatrix}^{-1}[\mbf{k}^\top_n, 0]^\top = \sigma_n^2,
    \end{split}
\end{equation}
where we have assumed, without loss of generality, that $k((\x,t), (\x,t))=1$ and have used the matrix inversion lemma.
Therefore, writing our acquisition function using the reparameterization trick, we see that
\begin{equation}
    \begin{split}
    \alpha_2(\x^t_{n+1}) =& \mbb{E}_n \left[\max_{\tilde{\x}\in\X}{\mbb{E}_{n, 1}~h(y(\tilde{\x}, T))} \right]\\ =& \mbb{E}_n \left[\max_{\tilde{\x}\in\X}{\mbb{E}_\gamma~h(\mu_{n,1}(\tilde{\x}, T) + \sigma_{n,1}(\tilde{\x}, T)\gamma)} \right]\\ \approx& 
    \mbb{E}_n \left[\max_{\tilde{\x}\in\X}{\mbb{E}_\gamma~h(\mu_{n}(\tilde{\x}, T) + \sigma_{n}(\tilde{\x}, T)\gamma)} \right]\\
    =& \max_{\tilde{\x}\in\X}{\upsilon_n(\tilde{\x}, T)},
    \end{split}
\end{equation}
and thus, in the limit $|T - t_{n+1}| \rightarrow \infty$,
$\alpha_2(\x^t_{n+1}) = \max_{\tilde{\x}\in\X}{\upsilon_n(\tilde{\x}, T)}$.
\end{proof}
\label{prop:t_llt_T}

\begin{remark}
    Following \Cref{prop:t_llt_T}, one can  see that if $k_t(t_{n+1}, T) \approx 0$, then $k_t(t_i, T) \approx 0,~\forall i=1,\ldots, n$ since $t_{i} > t_{i-1}$. This further leads to $\mbf{k}_{n+1} \approx \mbf{0}$ and hence $\mu_{n,1} \approx 0$ and $\sigma_{n,1}^2 \approx 1$. This results in (for $|T-t_{n+1}| \gg \theta_t$),
    \begin{equation}
    \begin{split}
    \alpha_{2}(\x^t_{n+1}) =& \mbb{E}_n \left[\max_{\tilde{\x}\in\X}{\mbb{E}_\gamma~h(\mu_{n}(\tilde{\x}, T) + \sigma_{n}(\tilde{\x}, T)\gamma)} \right] \\
    \approx & \mbb{E}_n \left[\max_{\tilde{\x}\in\X}{\mbb{E}_\gamma~h(\gamma)} \right] \\
    =& c,
    \end{split}
    \label{e:2_inner}
\end{equation}
where $c$ is some constant. Suppose the $(n+1)$th decision is made by maximizing $\alpha_{2}(\x^t_{n+1})$ with a multistart local optimization procedure started from a set of starting points $\{\x\} \in \X$, where $\{\x\}$ are sampled from distribution $p_\x$. Then, in the limit $|T - t_{n+1}| \rightarrow \infty$, decisions made via our lookahead acquisition function are equivalent to samples drawn randomly from $p_\x$.\footnote{Note that this assumes that the local optimizer seeks a critical point and returns the starting point as the maximizer.} 

We take advantage of this property and set $p_\x = \mcl{U}(\X)$, where $\mcl{U}$ is the uniform distribution. In practice, this could result in a decisions that \emph{explore} $\X$ when the $T$ is far away and hence can lead to improved learning of $f(\x,t)$, thus facilitating a more accurate prediction of the maximum at $T$.
\end{remark}

\label{prop:special_cases2}
\end{proposition}

\section{Numerical experiments}
We demonstrate our proposed approach on synthetic test cases and a quantum optimal control problem. 

We compare our method with
widely used myopic approaches in BO, namely,  \texttt{EI}, \texttt{PI}, and \texttt{UCB}.
We refer to ``\texttt{EImumax}'' and ``\texttt{PImumax}'' as the \texttt{EI} and \texttt{PI} strategies,
respectively, with the target $\xi$ set as the maximum of the GP posterior mean at the
current step, that is, $\underset{\x \in \X}{\max}~ \mu(\x, t)$. We set 
the confidence parameter for \texttt{UCB}  to $\beta=2$. Additionally, we
 compare our method with ``\texttt{mumax},'' which selects
points by maximizing the GP posterior mean at the current time. We also make direct comparisons
between the myopic acquisition functions and their lookahead counterparts.
Moreover, we compare our method with a strategy that selects points uniformly at random
from $\x \sim \mcl{U}(\X)$; we call this acquisition function ``\texttt{Random}.'' 

Our
metric for comparison is the oracle value $f(\x_T, T)$ at the target time $T$. Each query to $f(\x, t)$ includes additive noise with mean zero and variance $\sigma_\epsilon^2 = 10^{-3}$. Each repetition of each experiment is provided with \emph{starting} samples that include oracle evaluations at a set of points distributed with uniform spacing in $[0, t_{\T{start}}]$ with each point sampled uniformly at random in $\X$, where $t_\T{start}<T$ is arbitrarily set for each experiment. Within each repetition the same starting samples are given to all acquisition functions. 


\subsection{Synthetic one-dimensional test cases}
\label{ss:synthetic_1d}
Our synthetic test functions take the
form $f(\x, t) = f_{\x}(\x) + f_{\x t}(\x, t)$ or  $f(\x, t) = f_{\T{ref}}(\x (t)) $ 
in order to  induce time dependence to the
maximizer of $f(\x,t)$.

We first consider quadratic one-dimensional functions $f_{\x}$
over the domain $\X = [0,1]$ with
\begin{equation} 
f_{\x}(\x) = -4 (\x - 0.5)^2.
\label{eq:1dq}
\end{equation}
The time-dependent component takes the following four forms:
\begin{equation*}
    \begin{split}
    \T{Quadratic-a}: \quad f_{\x t} = &\sin \left(\pi(\x+t)\right) + \cos \left(\pi(\x+t)\right) \\
    \T{Quadratic-b}: \quad f_{\x t} = &\sin \left(\pi(\x t)\right) + \T{cos}\left(\pi(\x t)\right) \\
    \T{Quadratic-c}: \quad f_{\x t} = &\sin \left(\pi(\x[t-3]^+)\right) +
    \T{cos}\left(\pi(\x[t-3]^+)\right)\\
    \T{Quadratic-d}: \quad f_{\x t} = &2\x\T{sin}(t) - \T{sin}^2(t),  
    \end{split}
\end{equation*}
where $[t-3]^+ \equiv \max(0, t-3)$ is specified to induce movement of
$\x_t^*$ for $t\geq 3$. 
\cref{f:1d_synthetic} shows the trajectories of $\x_t^*$, the maximizer at time $t$, for each of these test cases.
\begin{figure}
    \centering
    \includegraphics[trim={100 0 100 0}, clip, width=1\linewidth]{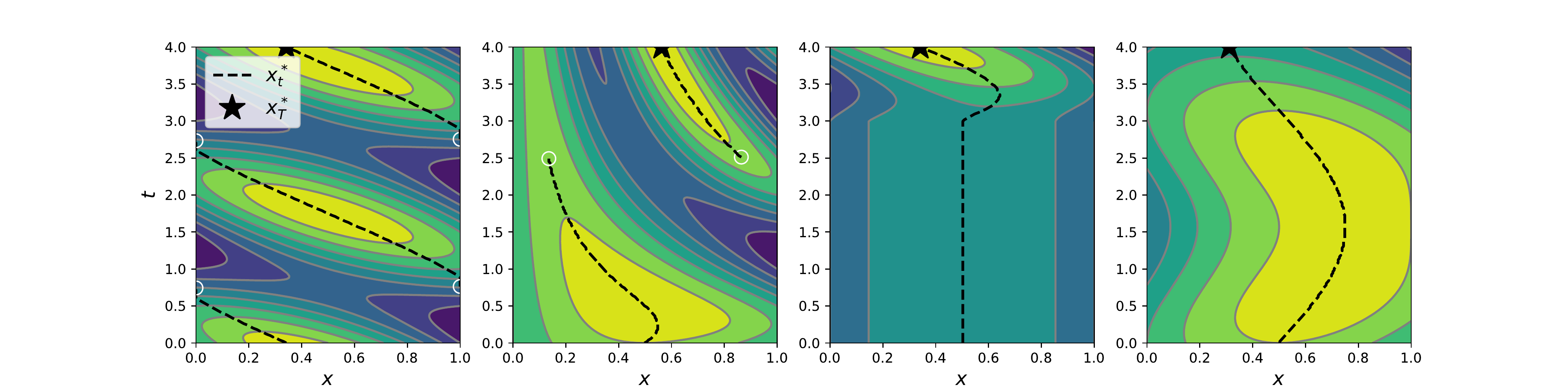}
    \caption{One-dimensional synthetic test cases. The dashed line represents the trajectory of $\x_t^*$.}
    \label{f:1d_synthetic}
\end{figure}
These test cases are deliberately designed to reflect a variety of situations, such as discontinuous change of $\x^*_t$ (Quadratic-a and Quadratic-b) and sudden dynamics (Quadratic-c). 

\cref{f:quadratic_exp} compares the points selected by our lookahead approach with various value functions and the corresponding myopic approaches. In each plot
20 repetitions 
are overlaid, with samples in $\X\times[0, 1]$ used to start all of the methods;  white
circles denote the points chosen by each method at times $t>1$. 
The lookahead BO does not try to track a moving maximizer, which results in points being more spread out in a space-filling fashion than with myopic acquisition functions. 
This approach is particularly  beneficial in handling oracles whose maximizers can go through a discontinuous change (e.g., Quadratic-a and Quadratic-b).
In the case of Quadratic-c, we demonstrate (see further details in \cref{sf:c}) that our algorithm is able
to handle oracles where $\x_t^*$ may undergo sudden dynamics ($\approx t=3.0$ in
this case). 
For Quadratic-d, where the challenges of the other
one-dimensional cases do not exist, the algorithm performs much better in
predicting $\x_T^*$.
\begin{figure}
    \centering
    \begin{subfigure}{.5\textwidth}
        \centering
        \includegraphics[trim={30 0 50 0}, clip, width=1\linewidth]{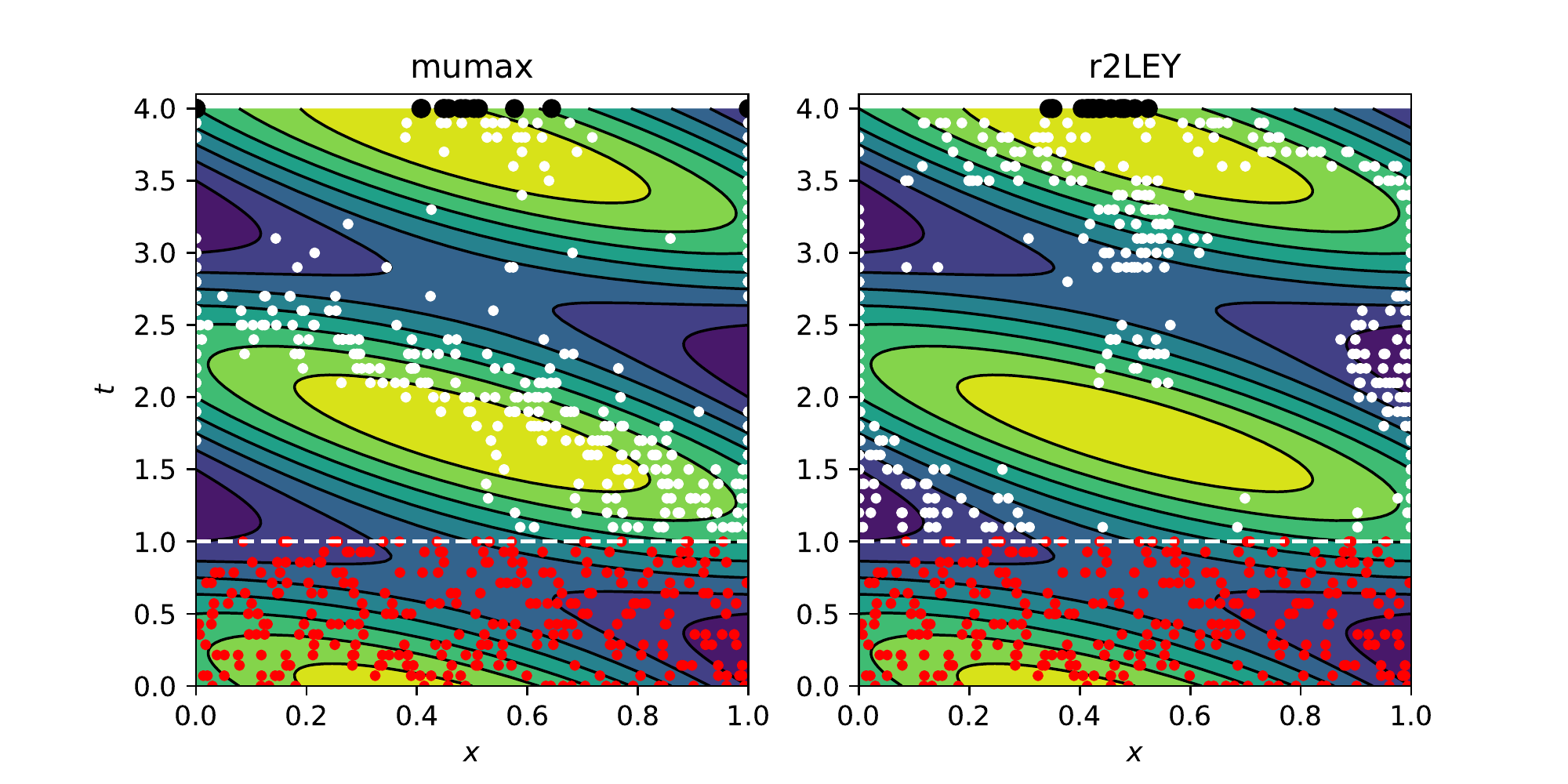}
        \caption{Quadratic-a}
    \end{subfigure}%
    \begin{subfigure}{0.5\textwidth}
        \centering
        \includegraphics[trim={30 0 50 0}, clip,width=1\linewidth]{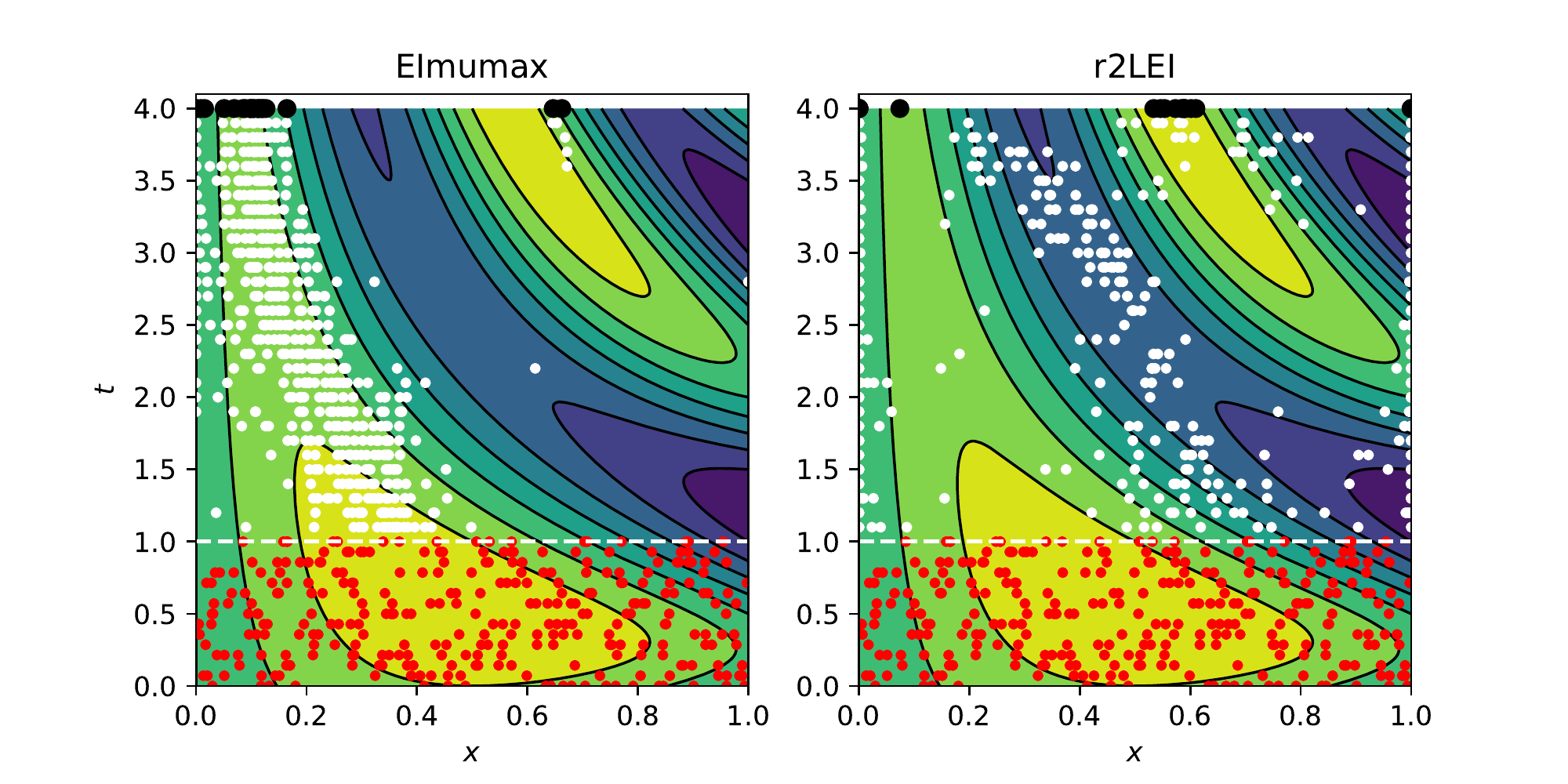}
        \caption{Quadratic-b}
    \end{subfigure} \\
    \begin{subfigure}{.5\textwidth}
        \centering
        \includegraphics[trim={30 0 50 0}, clip,width=1\linewidth]{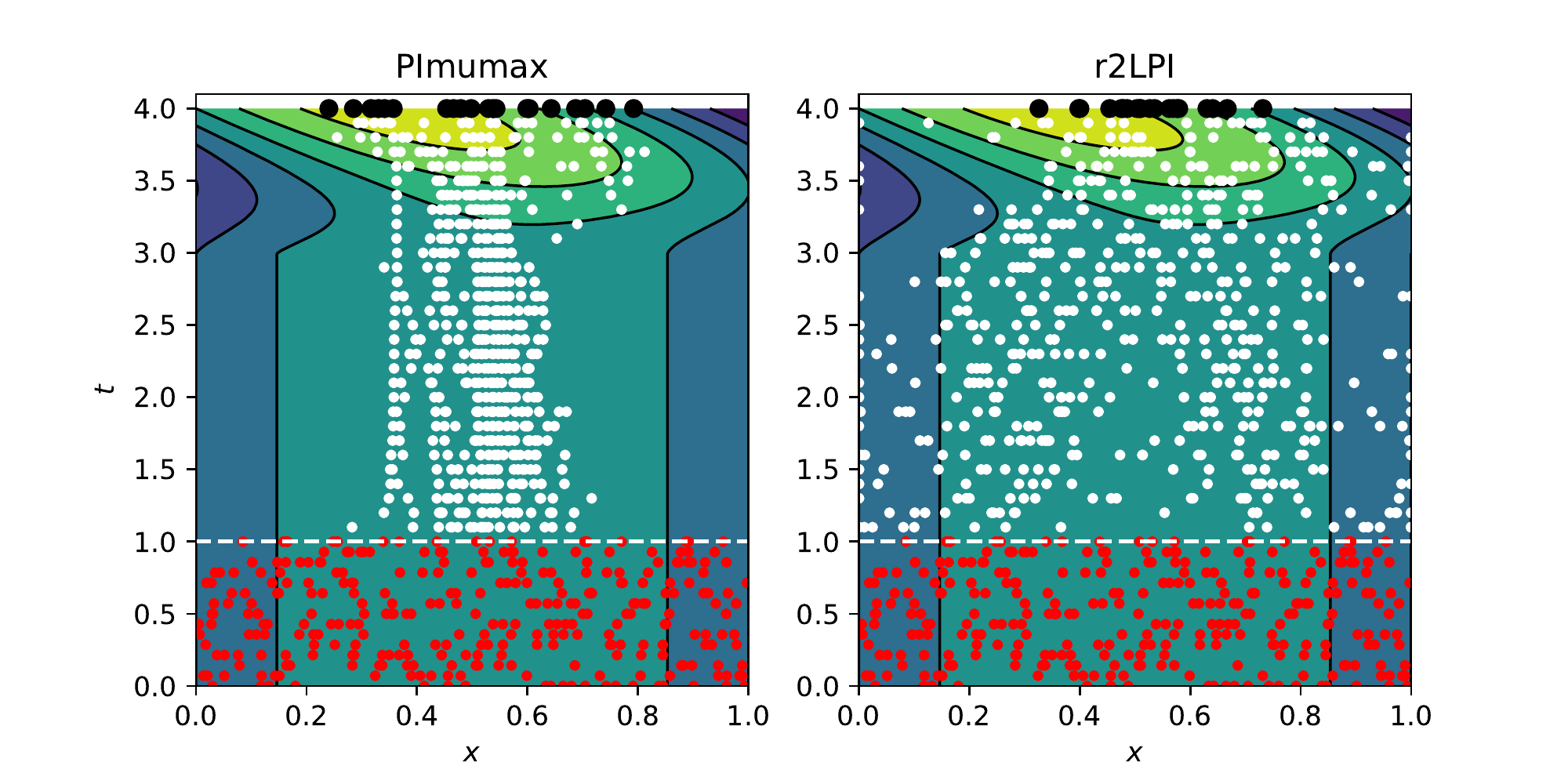}
        \caption{Quadratic-c}
    \end{subfigure}%
    \begin{subfigure}{.5\textwidth}
        \centering
        \includegraphics[trim={30 0 50 0}, clip,width=1\linewidth]{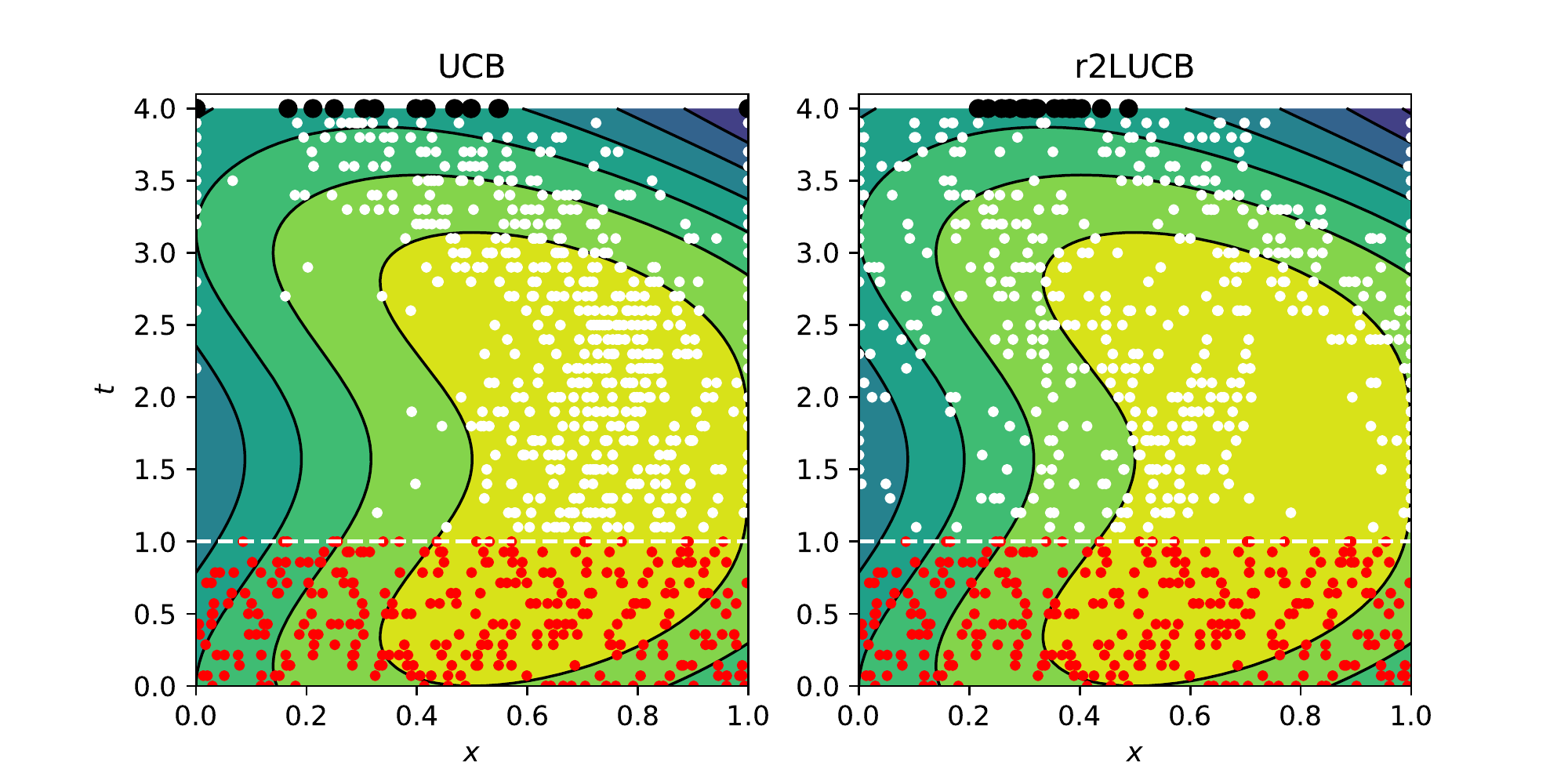}
        \caption{Quadratic-d}
    \end{subfigure}    
    \caption{Comparison of our two-step lookahead (right) and myopic (left) acquisition functions on the one-dimensional synthetic test cases. The plot shows
    an overlay of 20 repetitions of the algorithm with (shared) random starting points. The
    white circles denote the decision made at each step,
    and the contours in the background
    represents the true noise-free oracle $f(\x, t)$}
    \label{f:quadratic_exp}
\end{figure}

\cref{f:quad_myopic_v_lookahead} compares the time history of the average oracle value (over 20 replications
with randomized starting points) resulting from decisions made via the myopic and lookahead acquisition functions. We present results for only
Quadratic-a and Quadratic-c for the sake of brevity. One can see that
the lookahead acquisition functions (dashed line plots) make decisions with
lower oracle value during the early rounds. While the acquisition function per se does not
\emph{seek} to make decisions with low current oracle value, the consequence of
looking ahead at $T$ is that decisions made at current $t$ may incur a low oracle value. This is a fundamental distinction between our proposed lookahead and the myopic acquisition functions. Also, the average oracle value at target
time $T$ (denoted by the empty circle) is consistently higher than the myopic
counterparts (denoted by the square). Furthermore, the predictions by the
lookahead acquisition functions are with higher confidence, as visualized by
the error bars (one standard deviation in the predictions) shown at $T$; note that we highlight the error bars for the
lookahead acquisition functions in black.

The average oracle values at $T$
predicted by all the acquisition functions are somewhat consistent, with the
\texttt{r2LEY} being one of the most competitive, including the rest of the test cases.
For this reason, moving forward we predominately
use the \texttt{r2LEY} as the representative lookahead acquisition function.
\begin{figure}
    \centering
    \begin{subfigure}{1\textwidth}
        \centering
        \includegraphics[width=.8\linewidth]{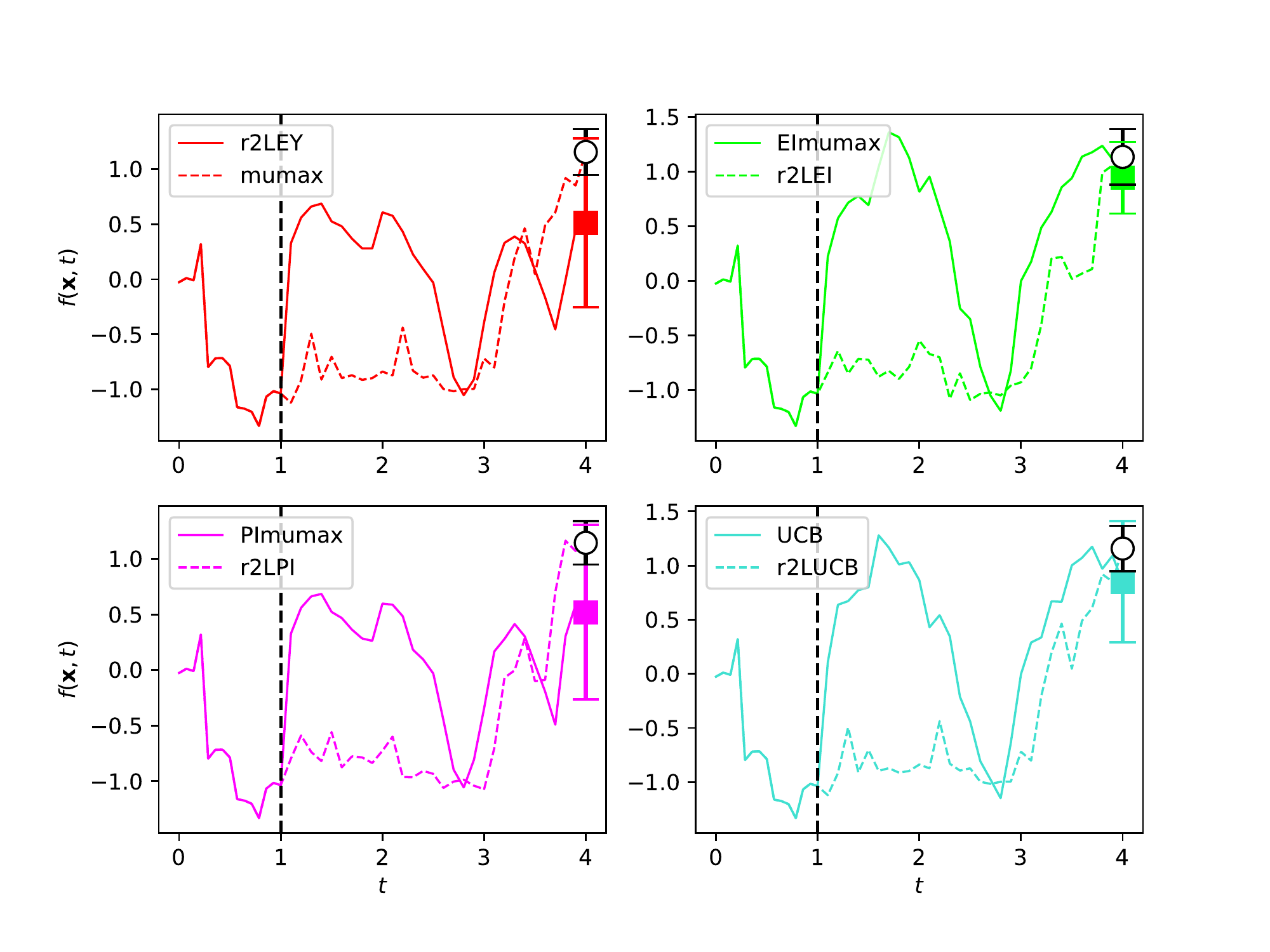}
        \caption{Quadratic-a}
        \label{sf:a}        
    \end{subfigure}\\
        \begin{subfigure}{1\textwidth}
        \centering
        \includegraphics[width=.8\linewidth]{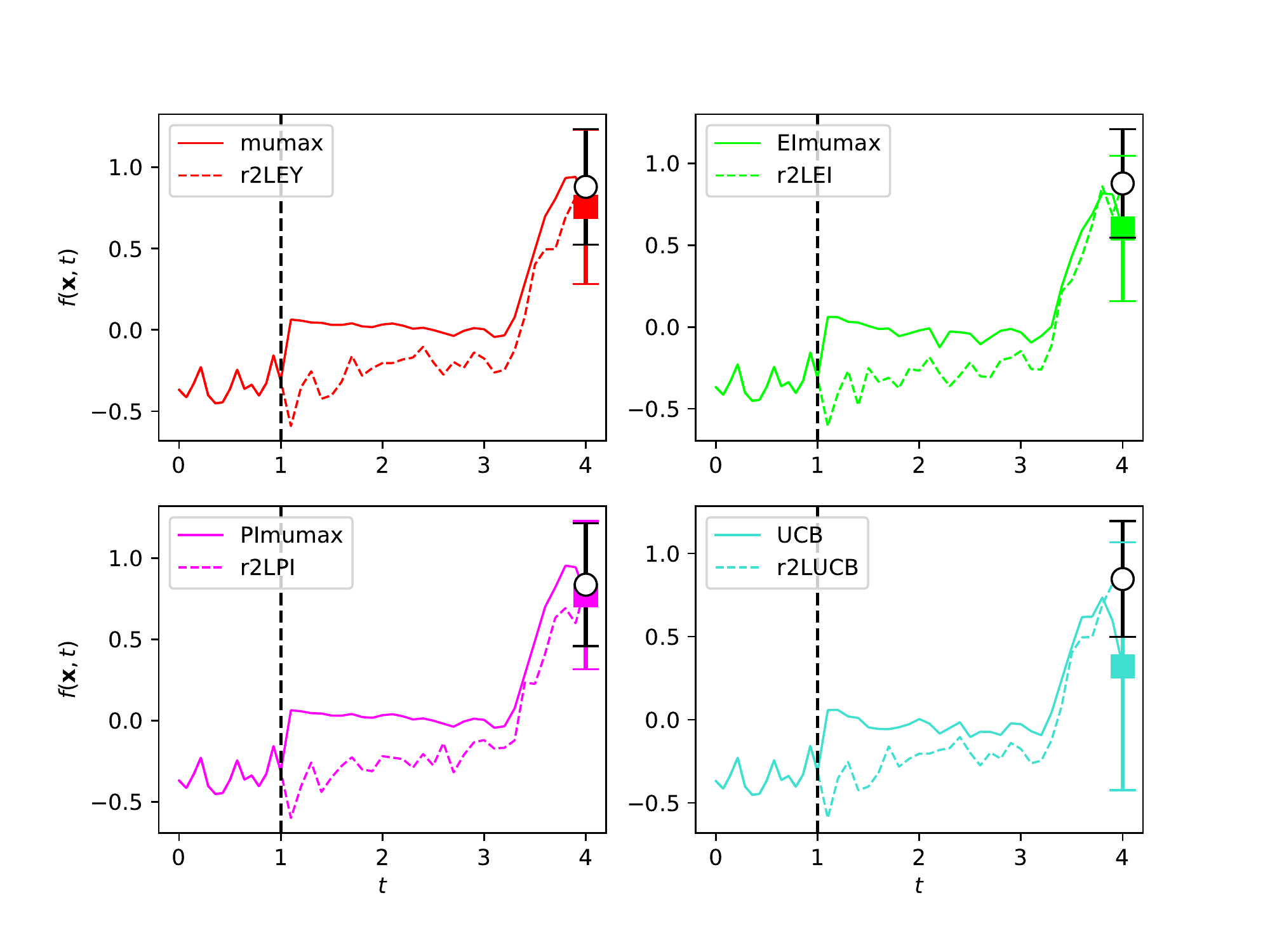}
        \caption{Quadratic-c}
        \label{sf:c}        
    \end{subfigure}
\caption{Average of 20 replications with random
starting points of the myopic (lookahead) acquisition function on 
one-dimensional test cases. The final point at $T=4$ is highlighted with square (circle) symbols and  variability at $T$ is visualized via
colored (black) error bars.
The vertical dashed line represents the time until which $n=15$ starting points
are collected and the total budget is set to $q=45$.}
\label{f:quad_myopic_v_lookahead}
\end{figure}


\subsection{Synthetic higher dimensional test cases}
\label{ss:synthetic_hd}
We now demonstrate our method on higher-dimensional test functions, up to $d=6$.
We begin by presenting a modified Griewank function
($d=2$)~\cite{simulationlib}, where we multiply the original
Griewank function $f_{\T{ref}}$ by a Gaussian weight function in order to create a
unique global maximum. Then, we modify the input $\x
\in \X=[-5,5]^2$ with a time-dependent rotation to yield the function
\begin{equation}
    f(\x,t) = f_{\T{ref}}(\x(t)) \times \T{exp}~\left( - \| \x - \x_0 \|^2/\ell\right), 
\end{equation}
where $\x_0 = [3, 0]$, $\ell \in \mbb{R}_+$  is a scaling
constant\footnote{We set
$\ell = 160$ to ensure the weighting is somewhat mild and does not
significantly change the nature of the true Griewank function.} and $\x(t) = R(\zeta_t) \times \x,$
where
\[ R(\zeta_t) = 
\begin{bmatrix}
\cos(\zeta_t) & -\sin(\zeta_t) \\
\sin(\zeta_t) & \cos(\zeta_t)
\end{bmatrix} \]
is a $2\times 2$ rotation matrix and $\zeta_t = \pi t /4$. Time snapshots of
this test function are shown in \cref{f:mod_griewank} and illustrate a 
counterclockwise rotation of the unique global maximizer with time. As before, all oracle observations include additive mean zero noise with variance $\sigma_\epsilon^2 = 10^{-3}$.

The performance of all the myopic acquisition functions is compared against \texttt{r2LEY} in \cref{f:mod_griewank_result}, where the
contour plot at $T$ is shown with the predictions from 20 repetitions of each
acquisition function (black $\times$ symbols) overlaid. Note that the total
budget is fixed at $q=90$, where $n=60$ is used in $\X \times [2,3]$ to
start the algorithm, and the target time is $T=4$.  The
lookahead acquisition function clearly predicts the global maximum at $T$ with greater
consistency than the rest do; all but two repetitions predicted $\x_T$ such that $\| \x_T - \x^*_T\| < 0.26$.

\begin{figure}[htb!]
    \centering
    \includegraphics[trim={50 50 50 50}, clip, width=.8\linewidth]{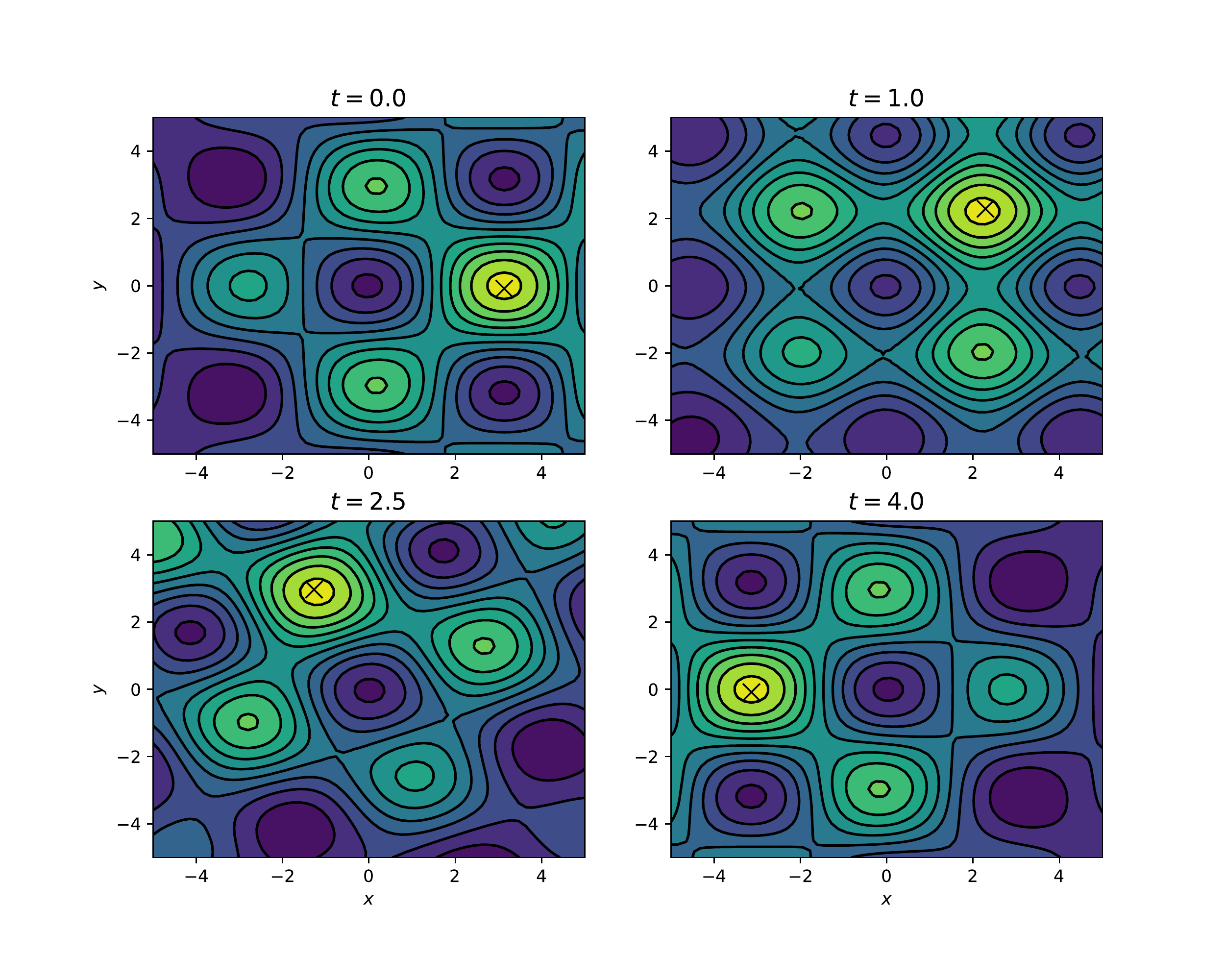}
    \caption{Time snapshots of the modified Griewank function. Notice the unique global maximum at each time and the counterclockwise rotation of the global maximizer with time.}
    \label{f:mod_griewank}
\end{figure}
\begin{figure}[htb!]
    \centering
    \includegraphics[width=.8\linewidth]{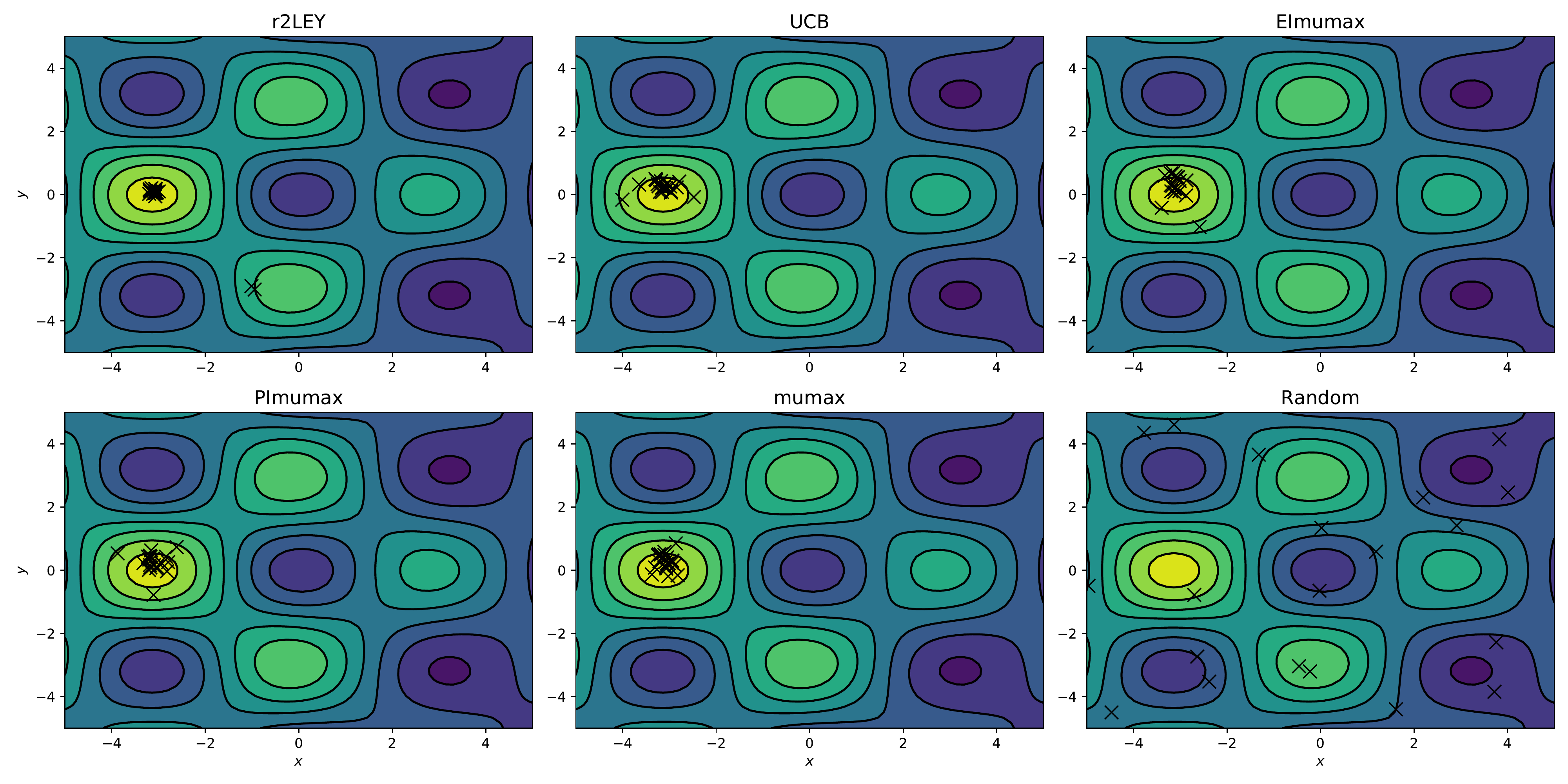}
    \caption{Predictions at $T=4.0$ on the modified Griewank test case. Symbols represent predictions at $T$ after 20 repetitions of each algorithm, and the contours represent the noise-free $f(\x, T)$. The total budget is set to $q=90$, where $n=60$ samples in $2.0\leq t \leq 3.0$ are used to start each algorithm.}
    \label{f:mod_griewank_result}
\end{figure}

We also demonstrate our method on the Hartmann-3d ($d=3$) and Hartmann-6d
($d=6$) test functions. In both cases, the Hartmann
function~\cite{simulationlib} is summed with $f_{\x t} (\x, t) = \mbf{1}^\top [2\T{sin}(t)\x
- \T{sin}^2(t)\mbf{1}]$ (where $\mbf{1}$ is a vector of ones of length $d$) to induce movement of the maximum with time.
For Hartmann-3d, the total budget is set to $q=120$ with $n=90$
samples queried in $[0, 1]^3 \times [2,3]$ to start the algorithm. For
Hartmann-6d, we set $q=150$ with $n=120$ samples queried in $[0, 1]^3 \times
[2,3]$ to start the algorithm, and our target time is $T=4$. 
\Cref{f:hartmann} shows the time histories of $f(\x_i,t)$
averaged over 20 replications. The predictions at $T$ are highlighted by adding
a marker in the plot. In both  test functions, the
lookahead acquisition function results in a higher oracle value $f(\x_T, T)$ at the
target time $T$ compared with other myopic acquisition functions. 

\begin{figure}[htb!]
    \centering
    \begin{subfigure}{0.5\textwidth}
        \centering
        \includegraphics[width=1\linewidth]{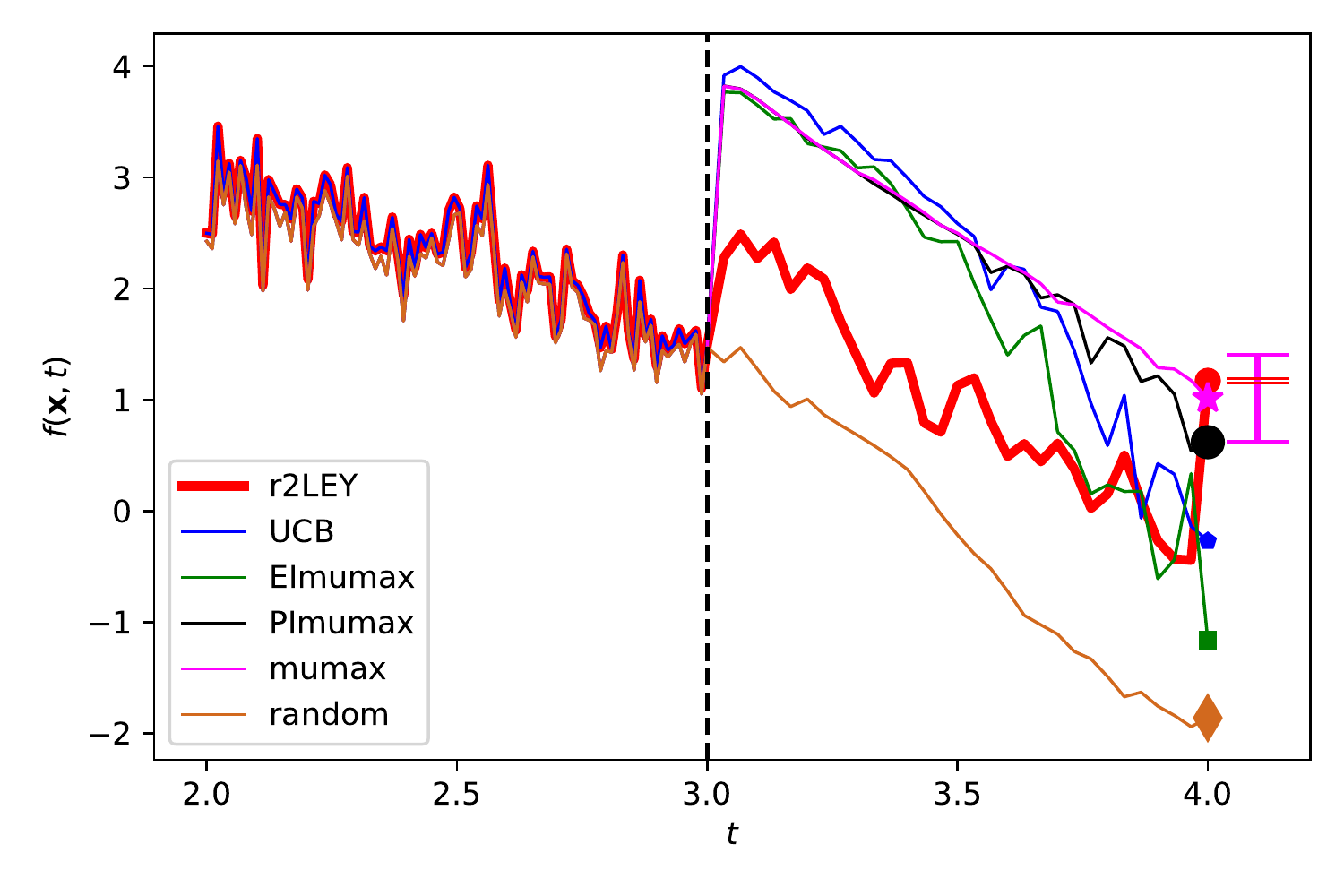}
    \end{subfigure}%
    \begin{subfigure}{0.5\textwidth}
        \centering
        \includegraphics[width=1\linewidth]{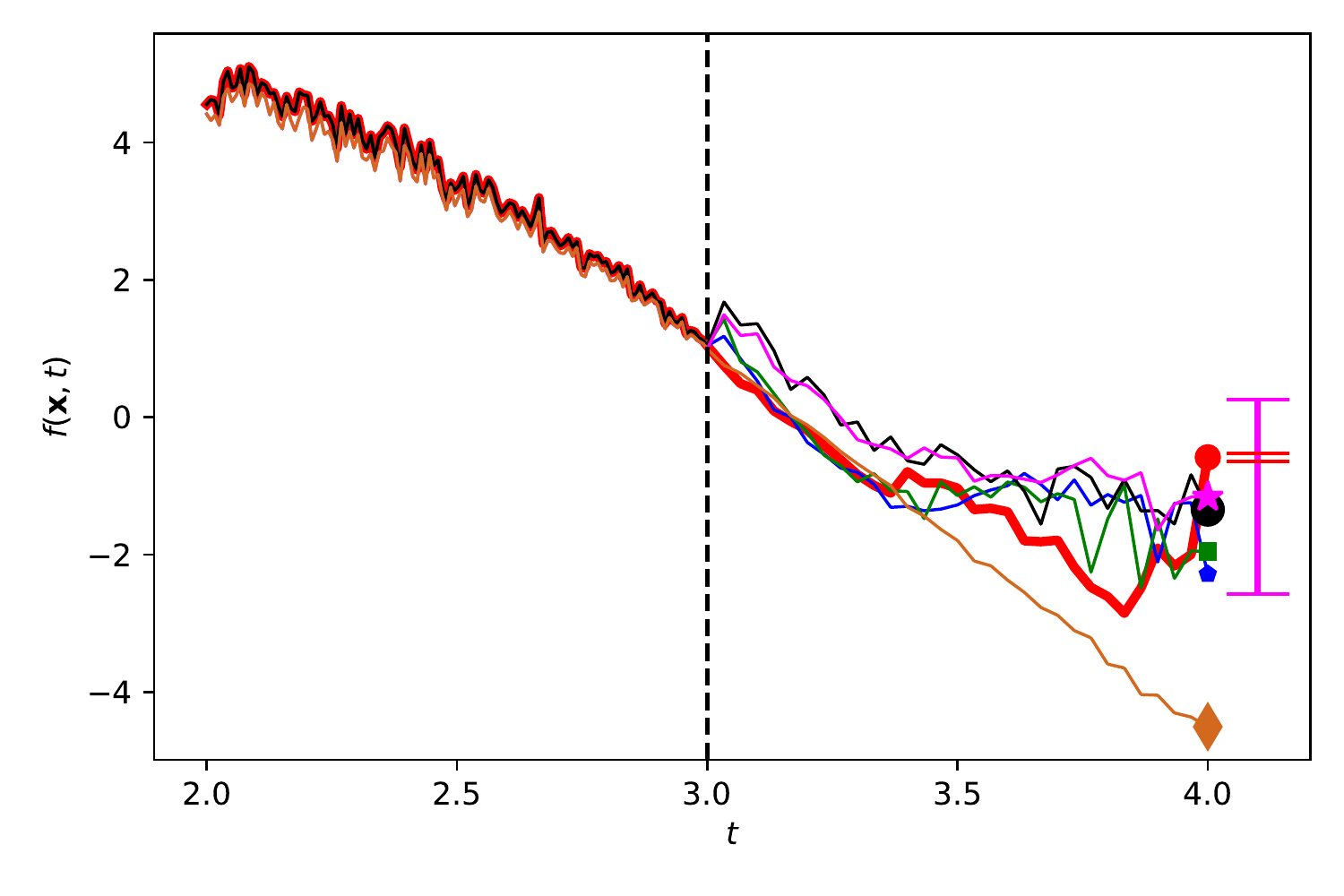}
    \end{subfigure}
    \caption{Hartmann-3d ($d=3$) and Hartmann-6d ($d=6$) test functions. Average of 20 repetitions. Left of vertical line corresponds to starting samples. The error bars for the oracle value at $T$ for the best two acquisition functions (\texttt{r2LEY} and \texttt{mumax}) shown are slightly shifted in time for better visibility. The \texttt{r2LEY} shows comparatively less variability.}
    \label{f:hartmann}
\end{figure}

\subsection{Quantum optimal control}
\label{ss:qoc}
We next compare methods for optimizing time-dependent oracles on a challenging real-world problem.
Optimally shaping electromagnetic fields to control quantum systems is a widespread
application~\cite{brif2010control}; similar problems appear when controlling molecular transformations relevant to chemical,
biological, and materials science applications. Quantum optimal control (QOC) is
a method to shape electromagnetic fields to steer a quantum system toward a
desired control target; see
\cref{f:quantum_chemistry} where the target is a specific outcome from a chemical reaction. 

\begin{figure}[tb!]
    \centering
    \includegraphics[width=0.5\linewidth]{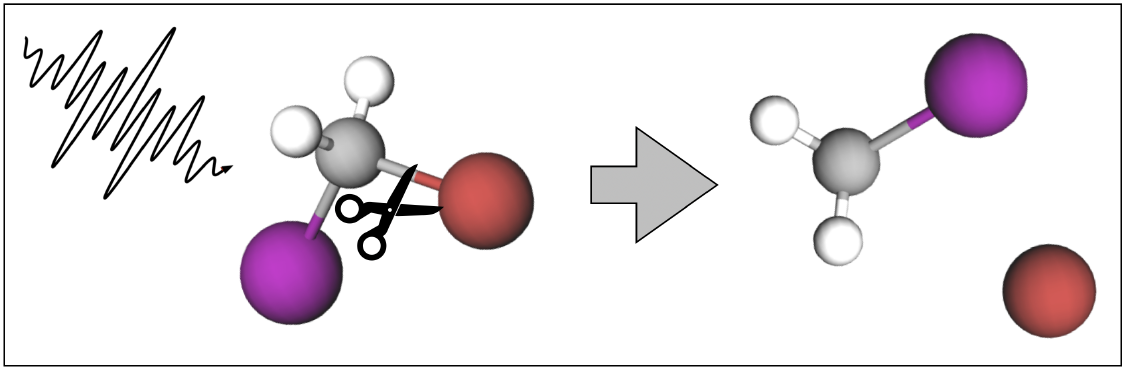}
    \caption{Illustration of QOC to control chemical reactions. A femtosecond
    laser controlled by QOC is used to provide control over the selective
    dissociation of molecules. Figure source:~\cite{magann2020digital}}
    \label{f:quantum_chemistry}
\end{figure}

The problem we consider is based on a diatomic molecule \ce{HF}, whose vibrational energy mode is modeled as a nonrotating Morse oscillator on a 1D grid that varies with time and is illuminated with a laser
field to induce dissociation. The electromagnetic field is parameterized by a set of coefficients $\x \in \X \subset \mbb{R}^5$ that are the control parameters; more specifically, $\x$ contains the frequency components of the electromagnetic field. The digital quantum simulation of the Morse 
oscillator $f(\x, t)$ takes as input the control parameter $\x$ and a \emph{terminal} time $t
\in [0, T]$---which is not to be confused with our target time $T$---whose
 output is the observable control objective that needs to be minimized with respect to $\x$. 
The cost of each query of the simulation increases monotonically with the terminal time $t$, and
we are interested in finding the optimal control at a specific target time $T$. Since searching for the optimal control at $T$ is expensive,  we seek to predict 
 the optimal control at $T$ with cheaper evaluations of $f(\x, t)$ for $t<T$.

In this regard, we use BO to find the control parameters $\x_T^*$ that maximize the negative value of $f(\x,T),
\forall \x \in \X$. We begin by collecting observations of their system
$\D_n= \{(\x_i, t_i), \hat{y}_i \},~i=1,\ldots,n$, which we use to fit
a GP that emulates $f$ in $\X \times {\mcl{T}}$. 
Then, we sequentially
make observations $\{(\x_i, t_i), \hat{y}_i \},~i=n+1,\ldots,q$, where 
each $\x_i$ is selected by maximizing our proposed two-step lookahead acquisition function. 
\Cref{f:quantum_problem} provides a schematic of
this problem. To conform to our problem setting, we allow evaluations only at increasing values of $t$. 
See~\cite{magann2020digital} for more information about this problem.

We set the total budget to $q=180$, where $n=150$ samples are queried over $4000
\leq t \leq 5000$ to start the algorithm and $T=6000$. The time histories of
each acquisition function for the QOC problem are shown in
\cref{f:quantum_results}, where we include results from \texttt{r2LEY} and \texttt{r2LPI} among the lookahead approaches. 
The best-known  maximum $f(\x^*,
T)=0.0$ at $T=6000$, which is shown in the figure as the black square. In terms of the average performance, the
lookahead \texttt{r2LEY} outperforms all other methods, with the \texttt{mumax}
being a  close competitor, as shown in the top-left of \Cref{f:quantum_results}. We visualize the variability in the predictions at $T$ via boxplots in the top right of \Cref{f:quantum_results}, where we observe that the \texttt{r2LEY} has the highest median oracle value to interquartile distance ratio. This shows that although \texttt{r2LEY} does not result in the best median oracle value, it provides the most confident prediction, while still being competitive with \texttt{UCB} and \texttt{mumax}. 
Even though the lookahead method \texttt{r2LPI} does not necessarily outperform all the other myopic methods, it shows better performance---in terms of both average and median oracle value---compared with its myopic counterpart \texttt{PImumax}. This behavior is consistent with  what was observed with the one-dimensional test cases presented in \cref{ss:synthetic_1d}.

\begin{figure}[tb!]
    \centering
    \includegraphics[width=1\linewidth]{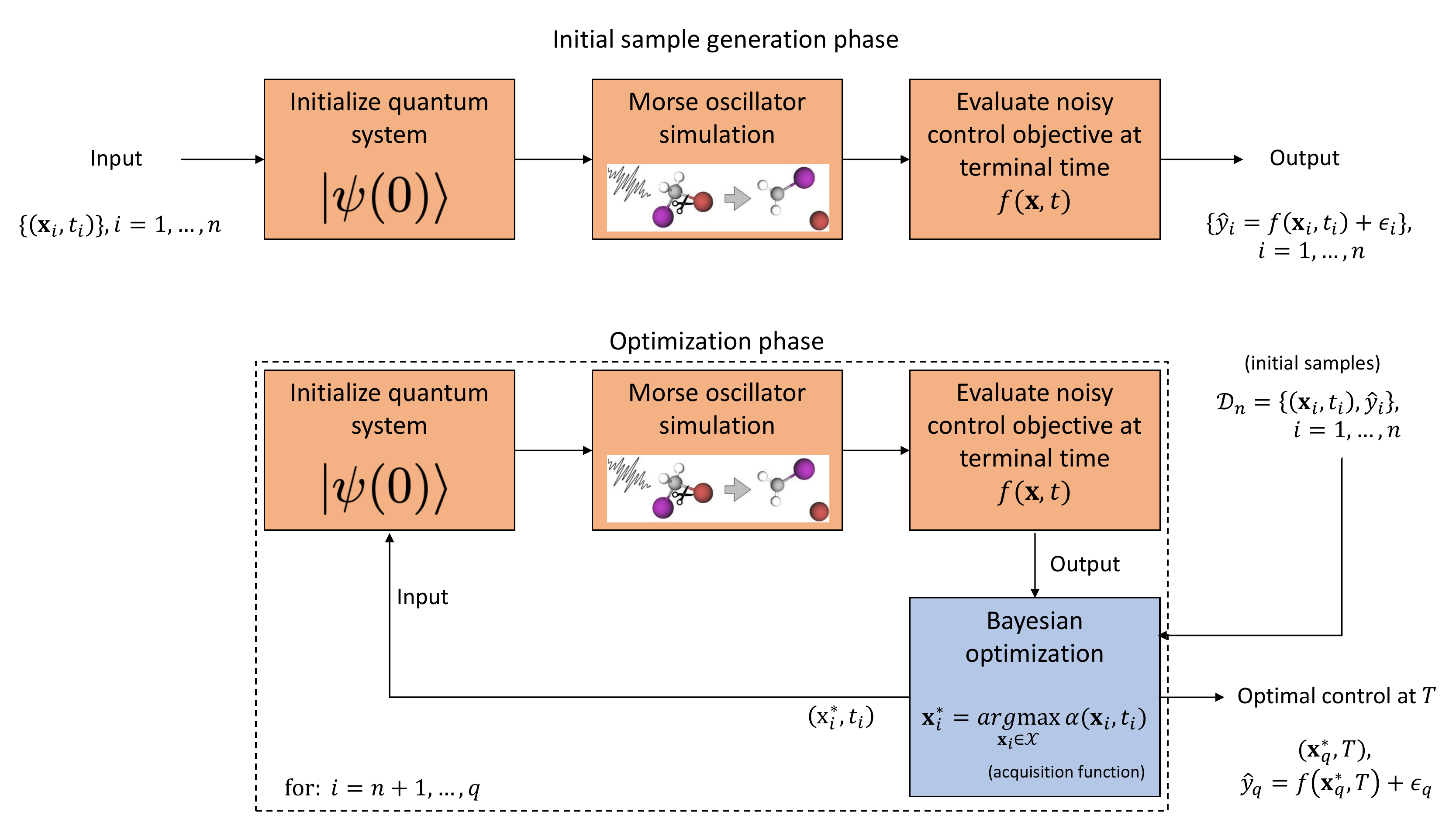}
    \caption{Overview of the quantum optimal control problem recreated from \cite{magann2020digital}.
    The quantum system dynamics, driven by a field
    parameterized by $\x$, are simulated for the time range $[0, t]$ (where $t$
    is the terminal time), after which the control objective function
    $f(\x,t)$ is calculated. A recent framework \cite{magann2020digital} is proposed
    for efficiently evaluating the control
    objective function 
    by using a future quantum computer.
    During an initial sample generation phase, the control objective function
    is evaluated $n$
    times to generate initial samples $\D_n$ that are used by the GP to
    learn $f(\x, t),~ \x\in \X, t\in [0, t]$. Then, in the optimization
    phase, the trained GP 
    is used in the BO framework to find the optimal
    control at target time $T$. The inputs at $(\x_i,t_i),~i=,n+1,\ldots,q$ are
    chosen at \emph{optimal} locations as guided by the BO acquisition function
    optimization.}
    \label{f:quantum_problem}
\end{figure}
\begin{figure}
    \centering
    \begin{subfigure}[t]{.5\textwidth}
        \centering
        \includegraphics[width=1\linewidth]{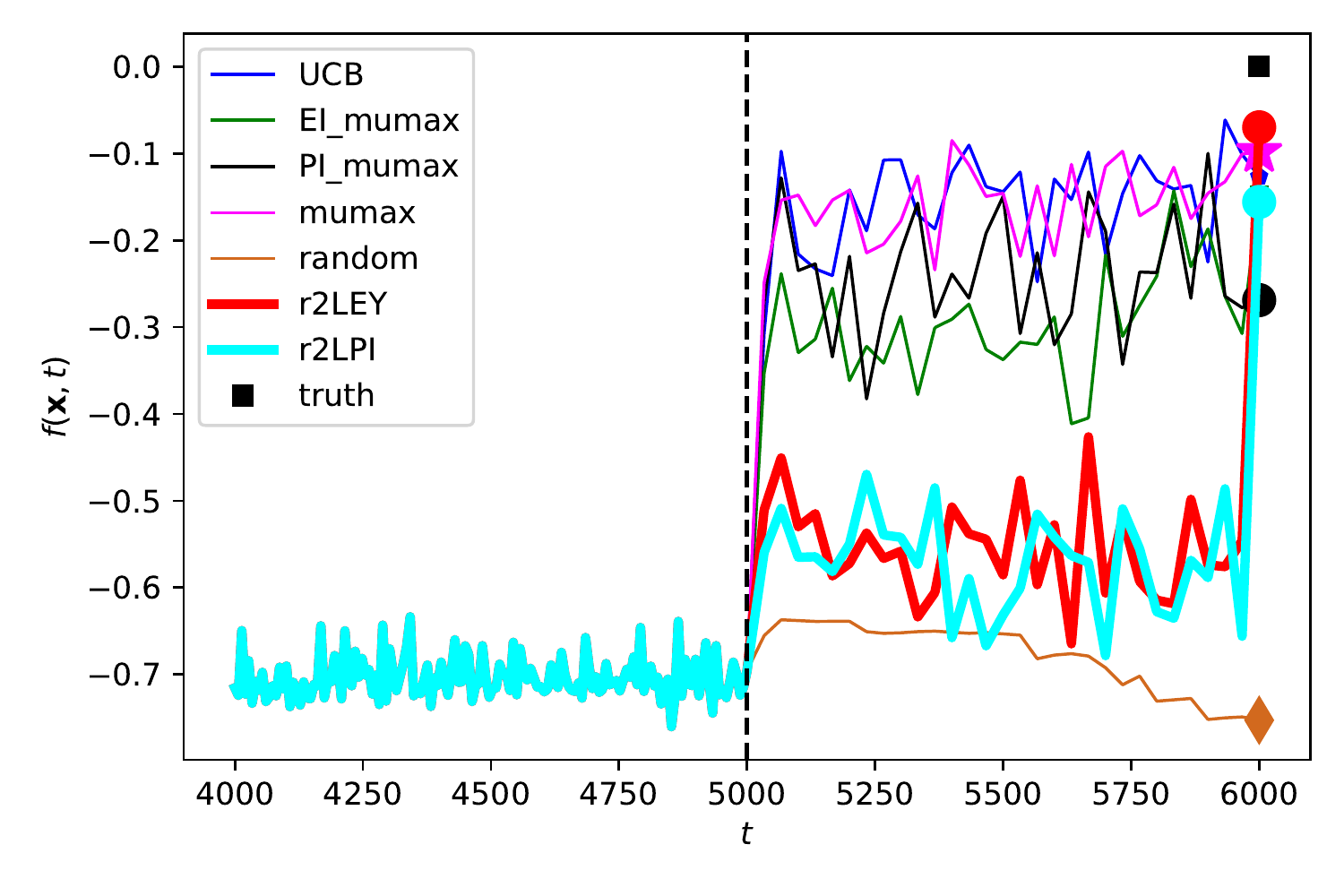}
        \caption{Average $f(\x,t)$ over 20 repetitions with randomized starting points. Left of vertical line denotes starting points.}
    \end{subfigure}%
        \begin{subfigure}[t]{.5\textwidth}
        \centering
        \includegraphics[width=1\linewidth]{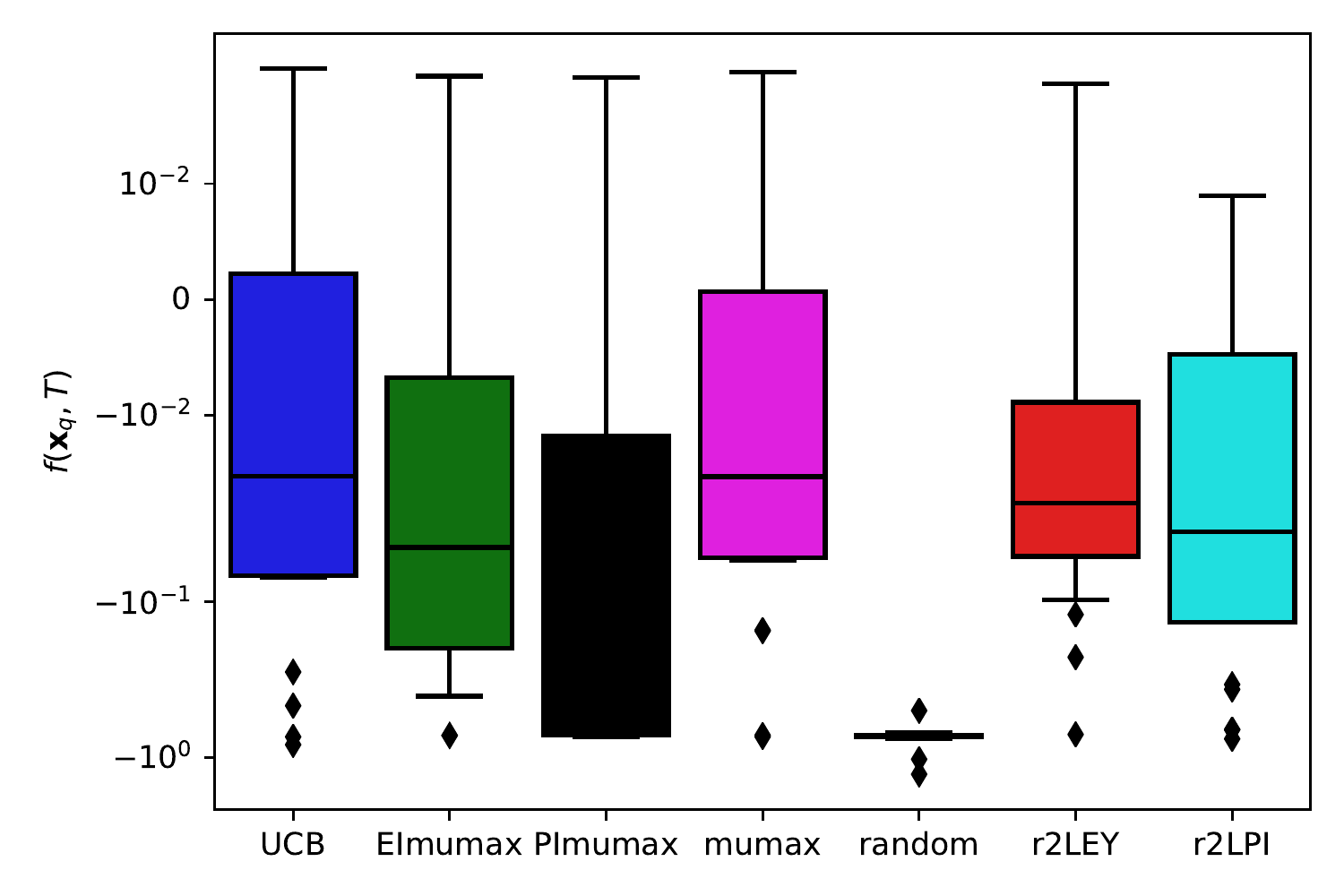}
        \caption{Variability in the oracle value at $T$, $f(\x, T)$. }
    \end{subfigure}\\

    \caption{Demonstration on the quantum optimal control problem.}
    \label{f:quantum_results}
\end{figure}

\subsection{Implementation details}
We implemented our recursive lookahead Bayesian optimization framework
in PyTorch~\cite{pytorch_bib} using the GP and BO libraries GPyTorch~\cite{gardner2018gpytorch} and
BoTorch, respectively~\cite{balandat2019botorch}. For
finite-horizon time-dependent BO with limited budget, the lookahead
acquisition functions outperform their myopic counterparts in securing the best
oracle value at the horizon. This improved accuracy comes at the cost of additional
computational overhead in terms of optimizing the acquisition functions. Our
Monte Carlo approximation to the acquisition function and its gradients, as presented in Line 5 of 
Algorithm~\ref{a:r2LBO}, lends
itself well to be used with a gradient-based optimizer.
For further computational improvements, however, we
also adopt the \emph{one-shot} optimization of the lookahead acquisition
function as implemented in BoTorch. The approach formulates line 4 of Algorithm~\ref{a:r2LBO} as a deterministic optimization problem over the joint $\x$ and $\x'$, where $\x' \defeq \{ \x^j\}_{j=1}^N$ are the candidate maximizers or \emph{fantasy} points. While this increases the dimensionality of the optimization problem from $d$ to $N+d$, we observe consistent results (with our Monte Carlo approach) with $N$ set to as small as 32, resulting in significant computational speedup. Moreover, the one-shot approach solves a deterministic optimization problem as opposed to a stochastic optimization problem with the Monte Carlo approach.
We include
both the Monte Carlo and one-shot approaches in our software and recommend the
latter as the dimensionality $d$ of the problem increases; all the results presented in this work are based on the one-shot approach. 
For the $d=5$
QOC problem, the acquisition function optimization takes approximately
30 seconds of wall-clock time when run in serial. Because many of the time-dependent oracles we are targeting require significantly more computational cost and time, we consider the per-iterate cost of our
approach to be trivial. Furthermore, the PyTorch framework allows
for leveraging advanced computation hardware such as graphics processing units, which may allow for further computational efficiency.

\section{Conclusions and future work}
Bayesian optimization with a lookahead acquisition
function outperforms well-known myopic acquisition functions in solving the finite-budget, finite-horizon, time-dependent maximization of an
expensive stochastic oracle.
Specifically, for such problems, the lookahead approach
makes probabilistically sound decisions that give better average-case
returns at the target horizon, while considering 
strategies that
lead up to the target horizon. In this regard, the lookahead acquisition
functions take advantage of the epistemic uncertainty---explained by the GP---in
the predictions about the oracle. In this work we introduce a generalized
framework that offers lookahead extensions to widely known myopic acquisition
functions. We demonstrate our method on 
illustrative 
synthetic test functions (that involve discontinuities in the
oracle maximizer) as well as a quantum optimal control application problem.
Because of the lack of analytically closed-form expressions for our lookahead
acquisition functions, their construction (and optimization) incurs more
computational overhead than do the myopic acquisition functions. Such an
overhead is insignificant, however, when optimizing expensive stochastic oracles that take, for example, hours of wall-clock time per evaluation. 
Overall, we demonstrate that our
proposed acquisition function results in more confident and accurate
predictions of the global maxima at the target horizon compared with Bayesian optimization using myopic acquisition functions.

The primary challenge with handling time-dependent optimization problems with
BO is that the effect of the dynamics of the oracle has to be learned by a GP model. When the
distance to the horizon is large, it is particularly challenging to train a
\emph{global} GP that emulates the oracle in $\X \times \mcl{T}$. A natural
alternative is to work with \emph{local} GPs within 
local
regions that
are dynamically updated. This is one avenue for future work that we
are considering.

A related problem is multifidelity
BO (see, e.g., \cite{song2019general, kandasamy2017multi, takeno2019multi}) in which one seeks to solve
$    \max_{\x \in \X }~f(\x)$ using observations of $\hat{f}(\x, s)$. The \emph{fidelity} is given by the
parameter $s\in[0,1]$, with $f(\x)\defeq \hat{f}(\x, 1)$ corresponding to the
highest fidelity. Typically, there is a cost function $c(s)\in \mbb{R}$, which is monotonically increasing in the fidelity
level $s$. 
By placing a GP prior on the function $\hat{f}(\x, s)$ directly, 
one can
sequentially optimize acquisition functions such as
\[ \max_{(\x,s) \in \X\times [0,1]}~\f{\alpha(\x, s)}{c(s)}.\]

Our problem is equivalent to the multifidelity BO problem with $s$
replaced by $t\in[0,T]$ with the following key differences:
\begin{enumerate}
    \item In our case, $\hat{f}(\x,T)$ can never be observed when $t< T$. More generally, any
      set $\{t_i \}, i=1, \ldots, n$ has a unique ordering $t_1 < \ldots < t_n$;
      and while we are at any $t_i$, evaluation of $f(\x, t \neq t_i)$ is not possible.
    \item We make only one observation $f(\x, t)$ for each $t\in[0,T]$.
    \item The cost of each evaluation $c(s) = c$ is a constant. However, similar to the general multifidelity framework, we
    are still interested in maximizing $f(\x, T)$.
\end{enumerate}
Another area of future work is defining a time-dependent cost
function that can be used within the multifidelity framework to also select the
time schedule of observations automatically; this approach  is of specific interest in
handling continuous-time systems.

\appendix
\section{Gradient computation with PyTorch} 
\label{app:grad_nd}
We show two examples of nondifferentiable functions in \Cref{f:nondiffgrad}:
(i) $f(\x) = \sin(|\x|)$ and (ii) $f(\x) = \max(\x, \x^2)$. In (i) the
nondifferentiability occurs at $\x=0$, and in (ii) the nondifferentiability
occurs at both $\x=0$ and $\x=1.0$, where the subgradient returned by PyTorch is shown in the figure.

\begin{figure}
    \centering
    \begin{subfigure}{.5\textwidth}
        \includegraphics[width=1\linewidth]{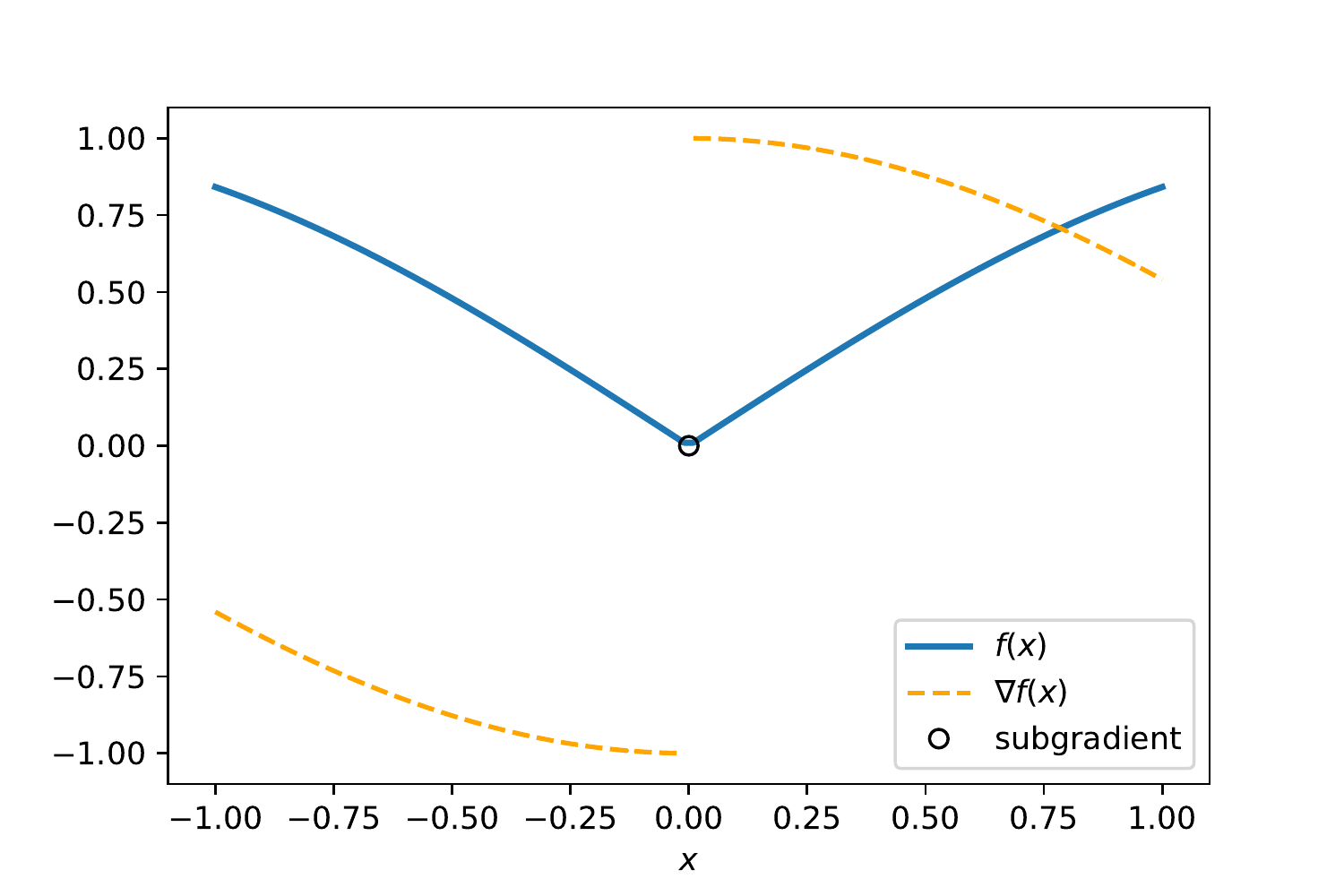}
        \caption{$f(\x) = \sin(|\x|)$}
    \end{subfigure}%
        \begin{subfigure}{.5\textwidth}
        \includegraphics[width=1\linewidth]{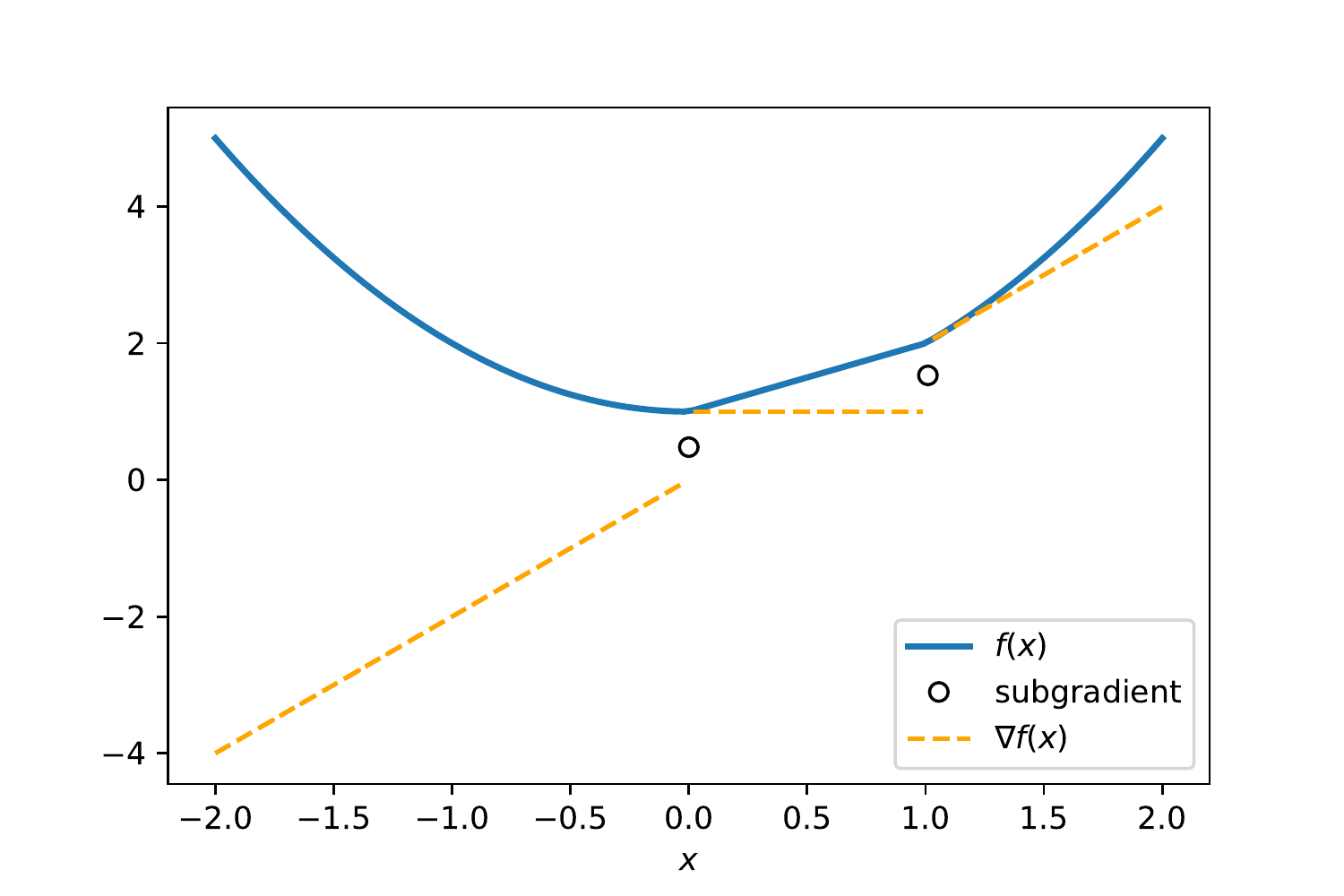}
        \caption{$f(\x) = \max(\x, \x^2)$}
    \end{subfigure}
    
    \caption{Examples of nondifferentiable functions and their (sub)gradients. At points of nondifferentiability, we indicate the subgradient returned by PyTorch.}
    \label{f:nondiffgrad}
\end{figure}


    



\section*{Acknowledgments}
This material was based upon work supported by the U.S. Department of Energy, Office of Science,
Advanced Scientific Computing Research, Accelerated Research for Quantum Computing and Quantum Algorithm Teams Programs under
Contract DE-AC02-06CH11357.
We thank Mohan Sarovar and Alicia Magann for
numerical simulation code for the Morse oscillator quantum optimal control
problem.

\begin{mdframed}
    The submitted manuscript has been created by UChicago Argonne, LLC, Operator of Argonne National Laboratory ("Argonne”). Argonne, a U.S. Department of Energy Office of Science laboratory, is operated under Contract No. DE-AC02-06CH11357. The U.S. Government retains for itself, and others acting on its behalf, a paid-up nonexclusive, irrevocable worldwide license in said article to reproduce, prepare derivative works, distribute copies to the public, and perform publicly and display publicly, by or on behalf of the Government. The Department of Energy will provide public access to these results of federally sponsored research in accordance with the DOE Public Access Plan (http://energy.gov/downloads/doe-public-access-plan).
\end{mdframed}
\bibliographystyle{siamplain}
\bibliography{references}

\end{document}

%% file: ex_shared_plain.tex
\usepackage{lipsum}
\usepackage{amsfonts}
\usepackage{graphicx}
\usepackage{epstopdf}
\usepackage{algorithmic}
\ifpdf
  \DeclareGraphicsExtensions{.eps,.pdf,.png,.jpg}
\else
  \DeclareGraphicsExtensions{.eps}
\fi

\usepackage{enumitem}
\setlist[enumerate]{leftmargin=.5in}
\setlist[itemize]{leftmargin=.5in}

\newcommand{\defeq}{:=}
\newcommand{\PI}{\textrm{PI}}
\newcommand{\EI}{\textrm{EI}}
\newcommand{\UCB}{\textrm{UCB}}

\title{Lookahead Acquisition Functions for Finite-Horizon Time-Dependent Bayesian Optimization and Application to Quantum Optimal Control}

\author{S. Ashwin Renganathan{\thanks{Mathematics \& Computer Science Division, Argonne National Laboratory, Lemont, IL 
  (\url{srenganathan@anl.gov}, \url{jmlarson@anl.gov}, \url{wild@anl.gov}).}}
  \and
  Jeffrey Larson\footnotemark[1]
  \and 
  Stefan M.\ Wild\footnotemark[1]
  }

\usepackage{amsopn}
